\documentclass[11pt]{article}
\usepackage{amsmath,amssymb,amsfonts,amsthm,epsfig}
\usepackage{cancel}

\usepackage[dvipsnames]{xcolor}
\usepackage{bm,xspace}
\usepackage{tcolorbox}
\usepackage{cancel}
\usepackage{fullpage}
\usepackage{pras}
\usepackage{framed}
\usepackage{verbatim}
\usepackage{enumitem}
\usepackage{array}
\usepackage{multirow}
\usepackage{afterpage}
\usepackage{mathrsfs}
\usepackage{pifont} 
\usepackage{chngpage}
\usepackage[normalem]{ulem}
\usepackage{boxedminipage}
\usepackage{caption}
\usepackage{tikz}
\usepackage{mdframed}
\usepackage{setspace}
\usepackage{multicol}
\usepackage{pgfplots}
\pgfplotsset{compat=1.18}

\newcommand*\dif{\mathop{}\!\mathrm{d}}

\def\anonymizeymize{0}

\ifnum\anonymizeymize=1
\newcommand{\anonymize}[1]{}

\fi

\ifnum\anonymizeymize=0
\newcommand{\anonymize}[1]{#1}

\fi

\definecolor{darkgreen}{rgb}{0.0, 0.5, 0.0}

\def\comments{0}

\ifnum\comments=1
\newcommand{\pras}[1]{\textcolor{BrickRed}{\sf{[#1 --PR]}}}
\newcommand{\kangning}[1]{\textcolor{orange}{\sf{[#1 --KW]}}}
\newcommand{\moses}[1]{\textcolor{darkgreen}{\sf{[#1 --MC]}}}

\newcommand{\alex}[1]{\textcolor{blue}{\sf{[#1 --AL]}}}
\newcommand{\edit}[2]{\textcolor{red}{#1}}

\fi

\ifnum\comments=0
\newcommand{\pras}[1]{}
\newcommand{\kangning}[1]{}
\newcommand{\moses}[1]{}
\newcommand{\alex}[1]{}
\newcommand{\edit}[2]{#1}
\fi

\def\colorful{1}

\ifnum\colorful=1

\fi
\ifnum\colorful=0

\fi

\DeclareMathOperator{\rank}{rank}

\newcommand{\cg}{\succ}
\newcommand{\cl}{\prec}
\newcommand{\cgeq}{\succeq}

\renewcommand{\hat}[1]{\widehat{#1}}

\renewcommand{\sf}[1]{\textsf{#1}}

\DeclarePairedDelimiter{\ceil}{\lceil}{\rceil}

\makeatletter
\newtheorem*{rep@theorem}{\rep@title}
\newcommand{\newreptheorem}[2]{
\newenvironment{rep#1}[1]{
 \def\rep@title{#2 \ref{##1}}
 \begin{rep@theorem}\itshape}
 {\end{rep@theorem}}}
\makeatother
\theoremstyle{plain}

\newreptheorem{theorem}{Theorem}

\makeatletter
\newtheorem*{rep@claim}{\rep@title}
\newcommand{\newrepclaim}[2]{
\newenvironment{rep#1}[1]{
 \def\rep@title{#2 \ref{##1}}
 \begin{rep@claim}\itshape}
 {\end{rep@claim}}}
\makeatother
\theoremstyle{plain}

\newrepclaim{claim}{Claim}

\newtheorem{Alg}{Algorithm}

\crefname{equation}{eq.}{eqs.}

\begin{document}

\definecolor{myblue}{rgb}{0.15, 0.1, 0.75}
\definecolor{mygreen}{rgb}{0.15, 0.55, 0.1}
\definecolor{mypink}{rgb}{0.75, 0.05, 0.55}

\hypersetup{
    linkcolor = mygreen,
    citecolor = myblue,
    urlcolor = mypink
}

\title{
Six Candidates Suffice to Win a Voter Majority \vspace{8pt}
}

\author{
\anonymize{
Moses Charikar\\\hspace{0pt}{{\sl Stanford University}}
\and Alexandra Lassota \\\hspace{0pt}{{\sl Eindhoven University of Technology}}
\and Prasanna Ramakrishnan \\\hspace{0pt}{{\sl Stanford University}}
\and Adrian Vetta\\\hspace{0pt}{{\sl McGill University}}
\and Kangning Wang\\\hspace{0pt}{{\sl Rutgers University}}
}
}
\anonymize{
{\let\thefootnote\relax\footnotetext{Emails: \texttt{moses@cs.stanford.edu}, ~\texttt{a.a.lassota@tue.nl}, ~\texttt{pras1712@stanford.edu}, ~\texttt{adrian.vetta@mcgill.ca}, ~\texttt{kn.w@rutgers.edu}.}}
}

\date{}

\pagenumbering{gobble}
\thispagestyle{empty}
\maketitle

\begin{abstract}

A cornerstone of social choice theory is Condorcet's paradox which says that in an election where $n$ voters rank $m$ candidates it is possible that, no matter which candidate is declared the winner, a majority of voters would have preferred an alternative candidate. Instead, can we always choose a small \emph{committee} of winning candidates that is preferred to any alternative candidate by a majority of voters?

Elkind, Lang, and Saffidine raised this question and called such a committee a \emph{Condorcet winning set}. They showed that winning sets of size $2$ may not exist, but sets of size logarithmic in the number of candidates always do. In this work, we show that Condorcet winning sets of size $6$ always exist, regardless of the number of candidates or the number of voters. More generally, we show that if $\frac{\alpha}{1 - \ln \alpha} \geq \frac{2}{k + 1}$, then there always exists a committee of size $k$ such that less than an $\alpha$ fraction of the voters prefer an alternate candidate. These are the first nontrivial positive results that apply for all $k \geq 2$.

Our proof uses the probabilistic method and the minimax theorem, inspired by recent work on \emph{approximately stable committee selection}. We construct a distribution over committees that performs sufficiently well (when compared against any candidate on any small subset of the voters) so that this distribution must contain a committee with the desired property in its support.

\end{abstract}

\pagenumbering{arabic}

\section{Introduction}\label{sec:intro}
Voting is a versatile model for the aggregation of individual preferences to reach a collective decision. Disparate situations, such as constituents choosing representatives, organizations hiring employees, judges choosing prize winners, and even friends choosing games to play, can all be understood as a group of voters choosing from a pool of candidates. Voting theory seeks to understand how winning candidates can be selected in a fair and representative manner.

One of the longest known challenges with voting is \textit{Condorcet's paradox}, discovered by Nicolas de Condorcet around the French Revolution \cite{condorcet1785essai}.\footnote{It is plausible that in early academic explorations of voting, 13th-century philosopher Ramon Llull had already discovered the possibility of this paradoxical situation \cite{llull1274artifitium,hagele2000lulls}.} 
The paradox is that in an election where voters have ranked preferences over candidates, the preferences of the ``majority'' can be contradictory --- no matter which candidate is declared the winner, a majority of the voters would have preferred another candidate. In fact, the contradiction can be even more dramatic, with ``majority'' replaced by a fraction arbitrarily close to 1. An illustrative example is when the voters have cyclic preferences as, for example, displayed in \Cref{tab:cyclic}.

\begin{table}[h]
\centering
\begin{tabular}{cccccc}
\hline
$v_1$ & $v_2$ & $v_3$ & $v_4$ & $v_5$ & $v_6$\\
\hline
1 & 2 & 3 & 4 & 5 & 6\\
2 & 3 & 4 & 5 & 6 & 1\\
3 & 4 & 5 & 6 & 1 & 2\\
4 & 5 & 6 & 1 & 2 & 3\\
5 & 6 & 1 & 2 & 3 & 4\\
6 & 1 & 2 & 3 & 4 & 5\\
\hline
\end{tabular}
\caption{An election where voters have cyclic preferences. The column headed with $v_i$ represents the $i$th voter's ranking of the candidates (labeled $1, 2, \dots, 6$ from top to bottom). For each candidate, another candidate is preferred by 
every voter except one.}
\label{tab:cyclic}
\end{table}

Though it is impossible to always find a single candidate that is always preferred over the others by a majority (called a \textit{Condorcet winner}), one hope is that relaxations of this condition are still possible to achieve. A natural relaxation arises in the setting of \emph{committee selection}, where rather than choosing a single winner, the goal is to choose a \textit{committee} of $k$ winners. For example, a political system may have districts with multiple representatives, organizations may make many hires at once, and friends might play more than one game in an evening. Another view is that committee selection can be used as a filtering step in a process with more than one round, like primaries or runoffs, choosing interviewees for a position, or nominations for a prize.

In this context, Elkind, Lang, and Saffidine \cite{elkind2011choosing,elkind2015condorcet}
asked: is it always possible to find a small committee of candidates such that no other candidate is preferred by a majority of voters over each member of the committee? They called this committee-analogue of a Condorcet winner a \emph{Condorcet winning set}, and defined the \emph{Condorcet dimension} of an election as the size of its smallest Condorcet winning set. For example, the election depicted in \Cref{tab:cyclic} has Condorcet dimension 2, since any pair of diametrically opposite candidates such as $\{3, 6\}$ would be a Condorcet winning set. More generally, \cite{elkind2015condorcet} raised the following question for an arbitrary threshold of $\alpha$ in place of $\frac12$, and a target committee size $k$.

\begin{question}[\cite{elkind2015condorcet}]\label{q:main}
A committee $S$ is \textit{$\alpha$-undominated} if for all candidates $a \notin S$, less than an~$\alpha$~fraction of voters prefer $a$ over each member of $S$. 
For what values of $k \in \mathbb{Z}^+$ and $\alpha \in (0, 1]$ does every election have an $\alpha$-undominated committee of size $k$?
\end{question}

In particular, we would like to know, for each $k$, what is the smallest $\alpha$ for which $\alpha$-undominated committees of size $k$ always exist (and, equivalently, for each $\alpha$, the smallest $k$ such that these committees always exist). %

Condorcet's paradox (or rather, its aforementioned generalization) shows that for $k = 1$ and any $\alpha$ bounded away from 1, there are elections with no $\alpha$-undominated singleton candidates. 
For the threshold of $\alpha = \frac12$, \cite{elkind2015condorcet} constructed instances
with Condorcet dimension 3 
by taking the Kronecker product of two elections with cyclic preferences (see \Cref{tab:dim3}). This construction can be easily extended to give a lower bound of $\frac{2}{k + 1}$ on the smallest $\alpha$ such that there always exists an $\alpha$-undominated committee of size $k$ (see \Cref{sec:lbs}). They also showed that an election with $m$ candidates has Condorcet dimension at most $\ceil{\log_2 m}$; to see this, note that some candidate beats a majority of the other candidates, so we can iteratively add such a candidate to our committee and remove all the candidates that it beats. %

\subsection{Our Contributions} %

We prove that every election has Condorcet dimension at most 6. This result is a corollary of our main theorem, which gives a nontrivial existence result for $\alpha$-undominated committees of size $k \geq 2$. We note that the final result we prove (\Cref{thm:stronger}) is stronger, but we start with the approximation below as it is easier to get a handle on. (For a comparison, see \Cref{tab:best-alpha} and \Cref{fig:comparison-plot}.)

\begin{theorem}\label{thm:main}
If $\frac{\alpha}{1 - \ln \alpha} \geq \frac{2}{k + 1}$, then in any election, there exists an $\alpha$-undominated committee of size $k$.
\end{theorem}

For the specific threshold of $\alpha = \frac12$, \Cref{thm:main} applies as long as $k \geq 3 + 4 \ln 2 \approx 5.77$, and so any election has Condorcet 
dimension
at most $6$ (which is not far from the lower bound of $3$). Taking $k = 2$, \Cref{thm:main} implies that there always exists a pair of candidates such that no third candidate is preferred by more than roughly $80\%$ of the voters. Even replacing $80\%$ with $99\%$, this was previously unknown.
These results show that just by having a few winners instead of one, the most dramatic failures of Condorcet's paradox are avoidable. We emphasize that these results hold for \textit{any election}, regardless of the number of voters, the number of candidates, or the preferences that the voters have over candidates.

Our starting point for proving \Cref{thm:main} is the observation that \Cref{q:main} is closely linked to the problem of \emph{approximate stability} in committee selection \cite{jiang2020approximately}. The principle behind stability is that a subset of voters should have control over a subset of the committee of proportional size. That is, a committee of size $k$ is \emph{stable} (also referred to as \emph{in the core} \cite{scarf1967core,foley1970lindahl,DBLP:conf/sigecom/FainMN18}) if the fraction of voters that prefers any committee of size $k'$ is less than $\frac{k'}{k}$. We note that in this setting, voters have preferences over \textit{committees} rather than candidates. This more expressive space of preferences gives it the power to model a wide variety of preference structures, such as approval voting and participatory budgeting.

Unfortunately, in many settings, stable committees do not always exist. To remedy this, \cite{jiang2020approximately} put forth the following approximate notion of stability, and showed the surprising result that for any monotone preference structure and any $k$, a $32$-stable committee of size $k$ exists.

\begin{definition}[Approximately stable committees \cite{jiang2020approximately}]\label{def:approx-stable}
A committee $S$ of $k$ candidates is \textit{$c$-stable} if for any committee $S'$ of size $k'$, the fraction of voters that prefers $S'$ over $S$ is less than~$c\cdot \frac{k'}{k}$. 
\end{definition}

Consider the natural preference order over committees induced by rankings over candidates, where $v$ prefers $S'$ over $S$ if and only if she prefers her favorite candidate in $S'$ over her favorite in $S$. A simple observation (explained more fully in \Cref{sec:equiv}) shows that a committee of size $k$ is $c$-stable if and only if it is  $\frac{c}{k}$-undominated. For this ranked preference structure, the constant of $32$ in the result of \cite{jiang2020approximately} can be improved to $16$ using the existence of stable lotteries for these preferences \cite{DBLP:journals/teco/ChengJMW20}. Then, as a black box, \cite{jiang2020approximately} implies that $\frac{16}{k}$-undominated committees of size $k$ always exist, which in turn implies that we can always find Condorcet winning sets of size at most $32$. Since this conclusion follows easily from \cite{jiang2020approximately}, we attribute the first constant upper bound on the size of Condorcet winning sets to their work.

One can interpret the approximately stable committee problem as a version of \Cref{q:main} focused on the asymptotics of $\alpha$ as the committee size $k$ grows large.  For this purpose, \cite{jiang2020approximately} implies a result that is optimal up to a constant factor, but it says nothing nontrivial for committees of size at most $16$. %
In contrast, \Cref{thm:main} gives results even for $k = 2$, and outperforms the bound implied by \cite{jiang2020approximately} for $k \leq 1.75 \times 10^4$, despite only implying the existence of $O(\log k)$-stable committees. 

Nonetheless, we show that our techniques can be applied to the asymptotic setting as well, giving an improvement over \cite{jiang2020approximately}. 

\begin{theorem}\label{thm:stable}
In any election, there exists a $\frac{9.8217}{k}$-undominated committee of size $k$.
\end{theorem}

As a corollary, \Cref{thm:stable} implies the existence of $9.8217$-stable committees for preferences induced by rankings over candidates. We note that \Cref{thm:stable} improves \Cref{thm:main} for $k \geq 496$.

\subsection{Technical Overview}

Our approach, building on \cite{jiang2020approximately}, is to first construct a particular distribution over committees of size $k$, and then to show that by sampling from this distribution, the resulting committee is $\alpha$-undominated in expectation. 
In fact, \cite{elkind2015condorcet}'s proof that the existence of $O(\log m)$ size Condorcet winning committees in elections with $m$ candidates can also be viewed through this framework. There, we can consider the \emph{uniform} distribution over candidates. To construct the committee, we sample from this distribution, remove the candidates that are beaten, and recurse on the remaining candidates. In expectation, half of the candidates are removed in each round, so the algorithm is likely to end with a committee of $O(\log m)$ candidates. The greedy algorithm of choosing the candidate that beats the most others in each round can be viewed as derandomization via conditional expectation.

In this light, a natural approach to improving the $O(\log m)$ guarantee is to find a better distribution over committees. One of the insights in \cite{jiang2020approximately} was to construct this distribution via the equilibrium of a zero-sum game. In the game, the \textit{defender} chooses a committee $S$ of size $k$, and the \textit{attacker} chooses a candidate $a$. After the choices are made, the defender pays the attacker a dollar for each voter that prefers $a$ over all members of $S$.
The optimal strategy for the defender is to choose a committee randomly according to some distribution, which \cite{jiang2020approximately} call the \emph{stable lottery}. Stable lotteries generalize the well-studied \textit{maximal lotteries} \cite{kreweras1965aggregation,fishburn1984probabilistic}, which correspond exactly to the case where $k = 1$. 

Next, to create a committee of size $k$, \cite{jiang2020approximately} take a recursive approach. First, they sample a committee $S$ of size~$k/2$, and show that ignoring the 25\% of voters that least like $S$, any candidate $a$ is preferred over $S$ by less than a $\frac{8}{k}$ fraction of the voters (which are treated as an irrevocable loss). In the next step, they recurse on the ignored voters, sample a committee of size $k/4$, and lose less than another $\frac{4}{k}$ fraction of the voters against any candidate $a$. The committee size and fraction of voters we lose continue to decrease exponentially, and in the end we have a committee of size $k$ such that less than a $\frac{16}{k}$ fraction of voters prefer any candidate $a$.

\vspace{10pt}

To prove \Cref{thm:main}, we introduce three twists into this framework. Two are part of how we set up the zero-sum game in order to construct a distribution over committees that individual candidates perform poorly against (\Cref{lem:low}), and one is in how we show that in expectation, a random committee sampled from the distribution performs well (\Cref{lem:high} and \Cref{claim:high}).

\paragraph{Improving the game by confining the adversary.} First, we modify the setup of the game so that the adversary must choose both a candidate $a$ \emph{and} a subset $U$ of an $\alpha$ fraction of the voters. The adversary then only gets paid for the voters in $U$ that prefer $a$ over the committee $S$. By tying the hands of the adversary in this way, we can drive down the value of the game, which gives a more favorable guarantee for the distribution over committees.

Once we fix the distribution over committees (referred to by $\Delta$), we measure the quality of a candidate $a$ or committee $S$ with respect to a voter $v$ with a crucial notion that we call the \emph{rank}, denoted $\rank_v(a)$ or $\rank_v(S)$ (see \Cref{def:rank}). Roughly speaking, this is simply the probability when we sample from $\Delta$ that we get a committee that is worse than $a$ or $S$ in $v$'s preference.

\paragraph{The activation function.} The second twist is what we call the \emph{activation function} $g$, which gives us flexibility in how we measure each voter's preferences for candidates and committees. This function may seem somewhat enigmatic in the proof, but here we try to give some rough intuition for the idea behind it. The initial observation is that by using versions of the zero-sum game with different committee sizes, we can construct distributions over committees of size $k$ in a variety of ways. The simplest would be to take the optimal mixed strategy for the defender in the original game with committee size $k$, but we could also take the optimal strategy for size $k/2$, sample twice from it and take the union. These different ways of constructing the distributions can actually be interpreted as attaching different activation functions to the defender's distribution in the payoffs of the \emph{original} game. For example, sampling twice from the $k/2$ distribution is equivalent to attaching the function $g(x) = \sqrt{x}$, and the reason corresponds to the fact that sampling two uniform reals from $[0,1]$ and taking the max is equivalent to sampling one uniform real from $[0, 1]$ and taking the square root.
In the end, thanks to the generality of the minimax theorem, the proof works for any non-constant, non-decreasing function $g \colon [0, 1]\to \mathbb{R}_{\geq0}$ such that $g(x^k)$ is convex. These functions give a richer continuous space of options for modifying the game, some of which are not easily interpretable in terms of the intuition described above. Each choice of $g$ gives a different bound for $\alpha$ as a function of $k$, and so we can simply choose the function that gives the best guarantee. To establish \Cref{thm:main}, it turns out that the simple choice $g(x) = x^{1/k}$ is enough, but we can get a slightly stronger result (\Cref{thm:stronger}) with the best function.

\paragraph{A one-shot approach with finer accounting of all voters.} Third, we use a more precise approach for showing that some committee in the support of our distribution performs well, by diligently accounting for the contributions of each voter. In each step of \cite{jiang2020approximately}'s recursion, when they sample committee $S$, they consider for each voter $v$ whether or not $\rank_v(S)$ is above some threshold (called $\beta$, which is set to $\frac14$). The voters below the threshold are ignored, and then recursed on in the next iteration. There are two potential roadblocks in using this approach for small committee sizes. Intuitively, if we are aiming for a final committee size of around 6, the recursion cannot be very deep. Each iteration can only choose 2 or 3 candidates, for which the guarantee is insufficient. That is, the benefits of the recursion only kick in for sufficiently large committees, and for small committees, it is better to sample the whole committee in one shot (without recursion). 
Second, there is too much loss in evaluating each voter with a binary threshold, and without recursion, we need better accounting for voters with a low opinion of the committee. In \Cref{lem:high} and \Cref{claim:high}, we give a smoother analysis, which allows us to more precisely account for the contributions of each voter. To give some rough intuition, what we would like to show is that there is some $S$ such that the total sum of $\rank_v(S)$ is large for any subset of an $\alpha$~fraction of the voters. If we fix $S$ and plot each $\rank_v(S)$, ordered from smallest to largest, it suffices to bound the area under the bottom $\alpha$~fraction. It turns out that the worst case for these ranks (that minimizes the area for all $S$) is not a step function with a sharp threshold, but rather a linear function with slope $1$ (akin to the cyclic preferences depicted in \Cref{tab:cyclic}).

\vspace{10pt}

Finally, to prove \Cref{thm:stable}, we use our modifications in tandem with the recursive approach. In the proof of this theorem, the idea above that does the heavy lifting is the use of the activation function $g$.

\subsection{Related Work}

\paragraph{Proportionality in committee selection.}
In the context of committee selection, the principle of proportionality says that large voter coalitions should have their preferences fairly represented --- an idea that dates back to at least the 19th century \cite{droop1881methods}. Since its advent, a substantial body of research has been dedicated to studying the possibility and implications of proportionality.
One of the most widely studied models is \emph{approval voting}, where voters express their preferences by selecting a subset of candidates they approve of. We refer the reader to a survey by Lackner and Skowron \cite{lackner2023multi} for a detailed discussion on the topic. A key appeal of this model is that there are a wide variety of proportionality axioms such as \emph{justified representation} (JR) \cite{aziz2017justified} and its variants (for example, \cite{DBLP:conf/aaai/FernandezELGABS17,DBLP:conf/sigecom/BrillP23}) that are satisfied by natural rules (such as Proportional Approval Voting (Thiele’s rule) \cite{thiele1895om,kilgour2010approval,aziz2017justified,DBLP:conf/sigecom/PetersS20}, Phragm\'{e}n's rule \cite{phragmen1894methode,DBLP:conf/sigecom/PetersS20}, and the Method of Equal Shares \cite{DBLP:conf/sigecom/PetersS20,DBLP:conf/nips/PierczynskiSP21}). These ideas have also been impactful in practice, with for example, the historical use of Thiele's
rule and Phragm\'{e}n's rule \cite{DBLP:journals/corr/Janson16}, and the recent successful implementation of the Method of Equal Shares for participatory budgeting in several European cities \cite{MES}. Additionally, this line of work is driven forward by intriguing conjectures that even stronger axioms, such as core stability \cite{aziz2017justified,DBLP:conf/sigecom/FainMN18}, might be universally satisfiable as well.

In comparison, proportionality in committee selection with \emph{ranked} preferences is relatively under-explored. As \cite{lackner2023multi} suggest, part of the challenge is that notions of proportionality in the approval setting do not always generalize to the ranking setting. (Or, like with core stability, the analogous axioms are not always satisfiable \cite{DBLP:journals/teco/ChengJMW20}.) One particularly well studied class of rank-based committee selection rules is that of \emph{committee scoring rules} \cite{elkind2017properties}. These voting rules, which generalize scoring rules in the single-winner setting, capture several natural committee selection rules, and have been axiomatically characterized \cite{faliszewski2019committee, skowron2019axiomatic}. %
We refer the reader to \cite{faliszewski2017multiwinner} for a more in-depth discussion.

\paragraph{Committee analogues of Condorcet winners.} Grappling with Condorcet's paradox has been a major driving force in social choice theory, and naturally, other attempts have been made to extend the notion of Condorcet winners to the multi-winner setting. Fishburn \cite{fishburn1981majority,fishburn1981analysis} introduced the idea of a \emph{majority committee}, defined as a committee preferred by a majority of voters over any other committee of the same size. The \emph{Smith set} \cite{good1971note,smith1973aggregation} $S$ is defined as the minimal committee such that for any $a \notin S$ and $b \in S$, a majority of voters prefers $b$ over $a$. \emph{Uncovered sets} \cite{fishburn1977condorcet,miller1980new}, \emph{bipartisan sets} \cite{laffond1993bipartisan} (the support of \emph{maximal lotteries} \cite{kreweras1965aggregation,fishburn1984probabilistic}), and other tournament solutions \cite{brandt2016tournament} can also be viewed as multi-winner generalizations of Condorcet winners. However, these notions face the same challenge as Condorcet winners: they either do not always exist or sometimes coincide with the entire (potentially large) set of candidates. As in the single-winner case, the goal shifts to identifying Condorcet-consistent rules, which select a Condorcet winner (or the analogous multi-winner notion) when one exists \cite{coelho2005understanding,barbera2008choose}.

In this context, \Cref{thm:main} highlights a distinct advantage of the approach by Elkind, Lang, and Saffidine \cite{elkind2015condorcet}: small Condorcet-undominated sets are guaranteed to exist.

\paragraph{Other explorations of \Cref{q:main}.} Lastly, we mention a few other interesting explorations of Condorcet winning sets, and more generally $\alpha$-undominated sets. \cite{geist2014finding} used SAT solving to determine the largest Condorcet dimension in elections with a small number of voters and candidates. Their search did not turn up any instances with dimension larger than 3, but they found an election with just 6 voters and candidates with dimension 3, and they showed that this is the smallest possible. (We include one such instance in \Cref{tab:dim3-6}.) \cite{Bloks18} also explored whether elections with Condorcet dimension 4 could be constructed by exploring \textit{dominating sets} in tournament graphs, but that approach did not yield any such elections.
On the positive side, \cite{DBLP:journals/corr/LassotaVvS24} very recently showed that in elections where the voters and candidates are embedded in a 2-dimensional space, and their preferences are defined by their distance according to the $\ell_1$ or $\ell_\infty$ norm, the Condorcet dimension is at most 3.

In a more informal setting, a question isomorphic to \Cref{q:main} has also been explored in a series of MathOverflow posts \cite{mathoverflow-domination,mathoverflow-graph-domination,mathoverflow-circular-domination}, motivated by its connection to the \textit{discrete Voronoi game} \cite{DBLP:journals/jgaa/GerbnerMPPR14}. These posts offer an intriguing window into different approaches to resolving the problem, including why some natural approaches fall short. In their formulation \cite{mathoverflow-domination}, each candidate $a$ is represented by a function $f_a \colon [n] \to \mathbb{N}$, which can be thought of as a map from each voter $v$ to the rank of $a$ in $v$'s preference order. 
They ask \Cref{q:main}, with a particular focus on the case where $k = 2$. 
The responses contain examples of elections with Condorcet dimension 3, including the general lower bound that $\alpha$-undominated committees of size $k$ do not always exist when $\alpha < \frac{2}{k + 1}$~\cite{mathoverflow-lb}.

One natural approach towards positive results, suggested by Speyer \cite{mathoverflow-graph-domination}, is to solve the following graph theory question. 
\begin{question}\label{q:graph}
 For what choices of $\ell, k \in \mathbb{Z}^+$ does there exist a directed graph with girth larger than $\ell$ such that every subset of $k$ vertices has a common in-neighbor? 
\end{question}
 If there does \textit{not} exist such a graph for some choice of $\ell$ and $k$, then this implies that every election has a $(1 - \frac1\ell)$-undominated set of size $k$, by considering the graph on candidates where there is an edge $(a, b)$ whenever more than $1 - \frac1\ell$ fraction of the voters prefer $a$ over $b$. In particular, if every triangle-free graph has a pair of vertices without a common in-neighbor ($\ell = 3$ and $k = 2$), then this would imply that $\frac23$-undominated sets of size $2$ always exist, which would resolve \Cref{q:main} for $k = 2$. Unfortunately, such graphs \textit{do exist}. \cite{anbalagan2015large} gave a positive answer to \Cref{q:graph} for every  $\ell, k \in \mathbb{Z}^+$, using a construction based on additive combinatorics.

\section{Preliminaries and Notation}\label{sec:prelims}

\paragraph{Elections and preferences.} An \textit{election} (or \textit{preference profile}) is defined by the tuple $\mcal{E} = (V, C, \cg_V)$. $V$ denotes a set of $n$ voters, $C$ denotes a set of $m$ candidates, and $\cg_V$ denotes a set of linear orders (one per voter), representing the preferences of the voters over the candidates. We say that $a \cg_v b$ if $v$ prefers $a$ over $b$ (and $a \cgeq_v b$ if $a = b$ or $v$ prefers $a$ over $b$). We will use $a$ and $b$ to denote candidates, and in the body of the proof, we will use $a$ when the candidate is meant to be thought of as adversarially chosen and $\Delta_a$ to denote an adversarial \textit{distribution} over candidates.

We will use $S$ to denote a committee of at most $k$ candidates. When a voter $v$ prefers candidate~$a$ over each candidate in $S$, we write $a \cg_v S$. Here, $\tfrac1n |a \cg S|:= \frac1n\big|\{v\in V: a\cg_v S\}\big|$ denotes the fraction of voters that prefers $a$ over $S$. We note that if $a \in S$, then $\tfrac1n |a \cg S| = 0$.

With this notation, the central concepts of the paper can be defined as follows.

\begin{definition}[$\alpha$-undominated and $\alpha$-dominated]
A committee $S$ is \emph{$\alpha$-undominated} if for all $a \in C$, $\tfrac1n |a \cg S| < \alpha$. Conversely, $S$ is \emph{$\alpha$-dominated} by a candidate $a$ if $\tfrac1n |a \cg S| \geq \alpha$.
\end{definition}

We note that \cite{elkind2015condorcet} define equivalent notions in an opposite way. They define a \textit{$\theta$-winning} set, which is precisely a $(1 - \theta)$-undominated committee. The definition above is more in line with the approximate core stability definition, and we find it easier to think in terms of sets of voters that prefer $a$ over $S$, rather than sets of voters that prefer $S$ over $a$.

A \textit{Condorcet winning set} is a committee that is $\frac12$-undominated. The \textit{Condorcet dimension} of a fixed election is the size of its smallest Condorcet winning set.

For the purpose of proving \Cref{thm:main} and \Cref{thm:stable}, we assume that $\alpha n$ is an integer. To see why we can make this assumption, let us suppose we only prove these theorems when $\alpha n$ is an integer. The conditions of both of these theorems are of the form $\alpha \geq f(k)$. If $\alpha \geq f(k)$, then $\frac{\ceil{\alpha n}}{n} \geq f(k)$, and so there always exist $\frac{\ceil{\alpha n}}{n}$-undominated committees of size $k$. But a committee is $\frac{\ceil{\alpha n}}{n}$-undominated if and only if it is $\alpha$-undominated, so the results hold for any $\alpha \geq f(k)$.\footnote{Another observation that is helpful for intuition is that we can always create an arbitrarily large number of clones for each voter and candidate without changing the salient aspects of the problem.} 

As alluded to in the earlier discussion on stable committees, we will also impose a preference structure on committees induced by the voters' preferences over candidates. In particular, a voter $v$ \textit{strongly prefers} $S'$ over $S$ (denoted $S'\cg_v S$) if and only if $v$ prefers her favorite candidate in $S'$ over her favorite in $S$. If $v$ has the same favorite in $S$ and $S'$, we say that $v$ \textit{weakly prefers} $S'$ over $S$ (denoted $S'\cgeq_v S$) if $S \subset S'$ or $v$ prefers her favorite candidate in $S'\setminus S$ over her favorite in $S\setminus S'$.
It is not hard to check that these preferences define a complete and transitive order on the committees. Note that this distinction between strong and weak preferences is primarily important for approximately stable committees (\Cref{def:approx-stable}) in \Cref{sec:equiv}. Weak preferences are also used in \Cref{def:rank}, but any extension of strong preferences to a total order over committees with $\cgeq_v$ treated as ($\cg_v$ or $=$) would also suffice.

\paragraph{Distributions and the minimax theorem.} Throughout the paper, we use $\Delta$ to denote a distribution over committees of size at most $k$, and $\Delta_a$ to denote an adversarial distribution over candidates.

A central tool in our work is von Neumann's minimax theorem, stated below. 

\begin{theorem}[Minimax Theorem \cite{v1928theorie}]\label{thm:minimax}
Let $X \subset \mathbb{R}^p$ and $Y \subset \mathbb{R}^q$ be compact convex sets. If $f\colon X\times Y \to \mathbb{R}$ is a function such that $f(\bm{x}, \cdot)$ is convex for each fixed $\bm{x}\in X$ and  $f(\cdot, \bm{y})$ is concave for each fixed $\bm{y} \in Y$, then
$$\max_{\bm{x} \in X}\min_{\bm{y} \in Y} f(\bm{x}, \bm{y}) = \min_{\bm{y} \in Y}\max_{\bm{x} \in X} f(\bm{x}, \bm{y}).$$
\end{theorem}

A common way of applying the minimax theorem is in the context of two-player zero-sum games. In this setting, $X$ and $Y$ are the sets of probability distributions over actions that the two players can take (called \textit{mixed strategies}). For streamlined notation, we omit these domains when they are clear from context. For example, we may write $\displaystyle\min_{\Delta}$ to mean the minimum over all distributions over committees of at most $k$ candidates, and $\displaystyle\max_a$ to denote the maximum over candidates $a$.

We note that in most game-theoretic applications of the minimax theorem, the function $f$ is linear in all of its variables. (With a payoff matrix $A \in \mathbb{R}^{p\times q}$, $f(\bm{x}, \bm{y}) = \bm{x}^\top A \bm{y}$.) In our usage, we need the generality of the theorem, due to the non-linearity of the activation function $g$. %

\section{Proof of \texorpdfstring{\Cref{thm:main}}{Theorem 1}}\label{sec:mainthm}
We start by giving an overview of the structure of the proof and central concepts that are used throughout. Our goal is to prove the following theorem.

\begin{theorem}\label{thm:meat}
Let  $g\colon [0, 1] \to \mathbb{R}_{\geq0}$ be a non-constant, non-decreasing function such that $g(x^k)$ is convex. If $\alpha \in (0, 1]$ 
and $k \geq 1$ satisfy the condition  
$$\int_0^\alpha g(x)  \dif x \geq \int_{1 - \alpha}^1 g(x^k) \dif x$$
then in any profile, there exists a committee of $k$ candidates that is $\alpha$-undominated.
\end{theorem}

At a very high level, the two sides of the final condition are some measurement of the desirability of a committee and a candidate over subsets of an $\alpha$ fraction of the voters. The left side measures how much our committee $S$ is liked by the $\alpha$ fraction of voters that \textit{least} like it, whereas the right side measures how much any adversarial candidate is liked by the $\alpha$ fraction of voters that \textit{most} like it. If the left side is bigger than the right, then it is not possible for an $\alpha$ fraction of voters to prefer any $a$ over $S$.

Varying the function $g$ allows us to change the way that we measure the voters' preferences, for example, by prioritizing higher ranks and ignoring lower ranks. We write $g$ in an abstract way to show the flexibility of the proof techniques, but for the best intuition, we encourage the reader to think of $g$ as the identity function.

Before digging into the proof, first let us see why \Cref{thm:main} follows directly from \Cref{thm:meat}.

\begin{proof}[Proof of \Cref{thm:main} assuming \Cref{thm:meat}]
It suffices to show that the pre-condition of \Cref{thm:main} implies the pre-condition of \Cref{thm:meat}, since the conclusions of both theorems are the same. In particular, we will show that if $\frac{\alpha}{1 - \ln\alpha} \geq \frac{2}{k + 1}$, then $\int_0^\alpha g(x)  \dif x \geq \int_{1 - \alpha}^1 g(x^k) \dif x$ where $g(x) = x^{1/k}$. 

Suppose that $\frac{\alpha}{1 - \ln\alpha} \geq \frac{2}{k + 1}$. First, by rearranging, it follows that 
$$\frac{k}{k + 1} \left(1 + \frac{\ln\alpha}{k}\right) \geq 1 - \frac\alpha2.$$
Then using the fact that $\alpha^{\frac1k} = \exp\big(\frac1k\ln\alpha\big) \geq 1 + \frac{\ln\alpha}{k}$, we have
$$\frac{k}{k + 1} \alpha^{\frac1k} \geq 1 - \frac\alpha2.$$
Multiplying by $\alpha$, we have
$$\frac{k}{k + 1} \alpha^{1 + \frac1k} \geq  \frac12\alpha(2 - \alpha) = \frac12\big(1 - (1 - \alpha)^2\big).$$
The left and right sides are precisely $\int_0^\alpha x^{1/k}  \dif x$ and $\int_{1 - \alpha}^1 x \dif x$ respectively.
\end{proof}

Each possible choice of $g$ gives some tradeoff between $\alpha$ and $k$. Though $g(x) = x^{1/k}$ gives nearly the best result, ultimately it is not the optimal choice. In \Cref{ssec:opt-g}, we explore the small improvement that can be obtained with the optimal $g$.

\vspace{10pt}

 The techniques we use to prove \Cref{thm:meat} build on those developed by \cite{jiang2020approximately}. We proceed similarly via the probabilistic method. The first step of the proof  is to carefully design a distribution~$\Delta$ over committees of size at most $k$ which serves as a ``reference'' for measuring the quality of a candidate or a committee. The key property of $\Delta$ is that every candidate is viewed as ``low quality'' in comparison to committees drawn from $\Delta$ (\Cref{lem:low}, proved in \Cref{ssec:low}). The second step is to show that there exists a committee that is viewed as ``high quality'', and this is done by sampling from $\Delta$, and considering the quality of the committee in expectation (\Cref{lem:high}, proved in \Cref{ssec:high}). If a high-quality committee $S$ is $\alpha$-dominated by a candidate $a$, then $a$ must also be high-quality, which is contradiction. Our upper bound on the quality of candidates is precisely the right side of \Cref{thm:meat}, and the lower bound on the quality of some committee is the left side. 

We measure the quality of a candidate or committee on a voter-by-voter basis, using the following definition.

\begin{definition}\label{def:rank}
Given a distribution $\Delta$ over committees, we define the \emph{rank} of a candidate $a$ with respect to voter $v$ as 
$$\rank_v(a) := \Prx_{S'\sim \Delta}[a \cg_v S'].$$ 
Similarly, the rank of a committee $S$ with respect to voter $v$ is 
$$\rank_v(S) := \Prx_{S'\sim \Delta}[S \cgeq_v S'].$$ %
\end{definition}

The critical property of the rank is that if $a \cg_v S$, then $\rank_v(a) \geq \rank_v(S)$. We can visualize the ranks in the following way (see \Cref{fig:stacks}). Each voter has a column of height 1, divided into blocks. Each block corresponds to a committee $S$, and its height is the probability mass of $S$ in $\Delta$, denoted $\Prx_\Delta(S)$ (committees with probability 0 can be imagined as blocks with no height). The blocks are arranged in order of $v$'s preference, with the least preferred committees at the bottom and the most preferred at the top. 

\begin{figure}[ht!] %
  \centering
  \includegraphics[scale=0.899]{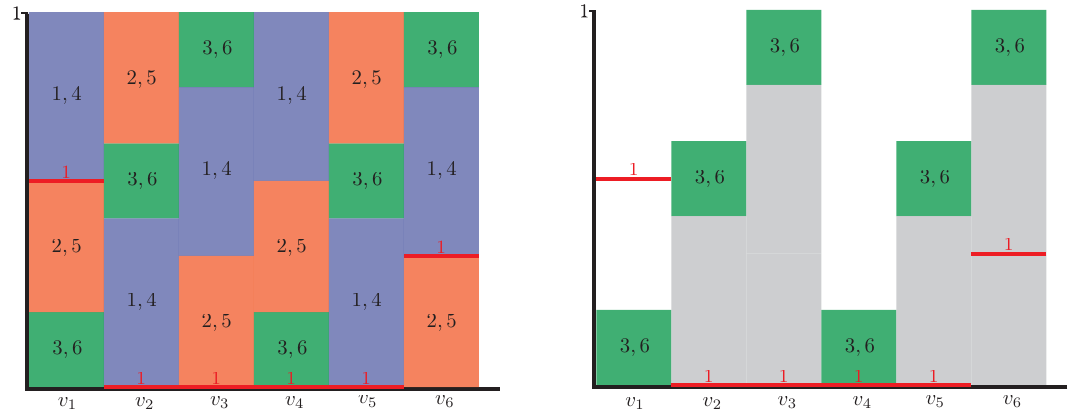} 
  \caption{A visualization of the ranks, in the election from \Cref{tab:cyclic}, and $\Delta$ which chooses the committees $\{1, 4\}, \{2, 5\}, \{3, 6\}$ with probabilities 0.45, 0.35, 0.2 respectively. The ranks of candidate $1$ are shown in red lines. One way to see that at most 2 voters prefer $1$ over $\{3, 6\}$ is that the three lowest ranks for $\{3, 6\}$ have a higher total than the three highest ranks for $1$. ($1.05$ from $v_1, v_2, v_4$ vs.\@ $0.9$ from $v_1, v_6,$ and any third.) %
  }
  \label{fig:stacks}
\end{figure}

The rank of a candidate is precisely the height of the bottom\footnote{The block for a candidate $a$ is the same as the block for the singleton committee $\{a\}$. These blocks typically have height 0, as depicted in \Cref{fig:stacks}, but ``bottom'' here is to be precise for any cases where $\Delta$ assigns positive mass to these committees.} of its block, and the rank of a committee is the height of the top of its block. One useful perspective that this visualization gives us is that from each voter $v$'s point of view, sampling from $\Delta$ is equivalent to sampling a real number $r$ uniformly at random from the interval $[0, 1]$, observing which of $v$'s blocks contains the point at height $r$, and selecting the corresponding committee (i.e., $S$ such that $\rank_v(S)$ is minimal, but at least $r$). %

With \Cref{def:rank} in hand, our notion of the ``quality'' of a candidate or committee can be thought of as the total \emph{area} under it, considered over subsets of $\alpha n$ voters. Our eventual goal is to prove the following two lemmas, which show that there is a distribution $\Delta$ such that all candidates occupy low ranks on average, and there exists a committee which occupies high ranks on average, over subsets of $\alpha n$ voters. In both lemmas, $g$ is an arbitrary non-constant, non-decreasing function such that $g(x^k)$ convex (in line with \Cref{thm:meat}).

\begin{lemma}[Candidates are low]\label{lem:low}
There exists a distribution $\Delta$ over committees of size at most $k$ such that for all candidates $a$, and all subsets of voters $U \subseteq V$ such that $|U| \leq \alpha n$, 
$$ \frac1n\sum_{v \in U} g(\rank_v(a)) < \int_{1 - \alpha}^1g(x^k)  \dif x.$$ 
\end{lemma}

\begin{lemma}[A committee is high]\label{lem:high}
For any distribution $\Delta$ over committees, there exists a committee $S$ such that for any subset of voters $U \subseteq V$ such that $|U| \geq \alpha n$, 
$$\frac1n\sum_{v \in U}g(\rank_v(S)) > \int_0^\alpha g(x)  \dif x.$$
\end{lemma}

It is not hard to see that together, these imply \Cref{thm:meat}.

\begin{proof}[Proof of \Cref{thm:meat} assuming \Cref{lem:low,lem:high}]
Suppose that, as in the pre-conditions of \Cref{thm:meat}, $g$ is a non-constant, non-decreasing function such that $g(x^k)$ is convex
and $\int_0^\alpha g(x)  \dif x \geq \int_{1 - \alpha}^1 g(x^k) \dif x$. 

Apply \Cref{lem:high} with the distribution given by \Cref{lem:low} to obtain a committee $S$. Suppose towards a contradiction that $a$ $\alpha$-dominates $S$. Then there exists a subset of voters $U \subseteq V$ with $|U| = \alpha n$ such that for all $v \in U$, $a\cg_v S$. For each $v \in U$, we have $\rank_v(a) \geq \rank_v(S)$. Since $g$ is non-decreasing, this implies that $g(\rank_v(a)) \geq g(\rank_v(S))$, and so
$$\frac1n\sum_{v \in U} g(\rank_v(a)) \geq \frac1n\sum_{v \in U}g(\rank_v(S)).$$
On the other hand, \Cref{lem:low,lem:high} give us the guarantee that  
$$\frac1n\sum_{v \in U} g(\rank_v(a)) < \int_{1 - \alpha}^1 g(x^k) \dif x \leq \int_0^\alpha g(x)  \dif x < \frac1n\sum_{v \in U}g(\rank_v(S))$$
which is a contradiction.
\end{proof}

\subsection{\texorpdfstring{\Cref{lem:low}}{Lemma 1}: Candidates are Low}\label{ssec:low}

The starting point for proving \Cref{lem:low} is the existence of a \emph{stable lottery}, stated below. Its existence is a corollary of \Cref{lem:low} in the case where $g(x) = x$ and $\alpha = 1$, but we will sketch its proof independently as a warm-up to proving \Cref{lem:low}.

\begin{corollary}[Existence of a stable lottery \cite{DBLP:journals/teco/ChengJMW20,jiang2020approximately}]\label{cor:stab}
There exists a distribution $\Delta$ over committees of size at most $k$ such that for all candidates $a$, 
$$\Ev_{S\sim \Delta} \left[\tfrac1n|a \cg S|\right] = \frac1n\sum_{v \in V} \rank_v(a) < \frac1{k+1}.$$
\end{corollary}

\cite{DBLP:journals/teco/ChengJMW20} originally proved this result, though technically with a slightly weaker upper bound of $\frac1k$. \cite{jiang2020approximately} later gave a simpler proof with the bound above. Both proofs use a game theory formulation and a clever application of the minimax theorem. We give a brief proof sketch in the following, as a warm-up for the proof of \Cref{lem:low}.

\begin{proof}[Proof sketch of \Cref{cor:stab} \cite{jiang2020approximately}]   
Imagine that we have a game with two players, the \emph{defender} who chooses a committee $S$ of size at most $k$, and the \emph{attacker} who chooses a single candidate $a$. Given their choices, the players receive payoffs $-\tfrac1n|a \cg S|$ and $+\tfrac1n|a \cg S|$ respectively. Our goal is to show that
$$\min_{\Delta} \max_{a} \Ev_{S\sim \Delta} \left[\tfrac1n|a \cg S|\right]  \leq \frac1{k + 1}.$$
i.e., the defender can choose a \textit{distribution} $\Delta$ over committees of size $k$ such that she can expect to pay at most $\frac1{k+1}$, regardless of the attacker's choice (which may depend on $\Delta$).

By the minimax theorem, we can swap the order of the players, as follows. 
$$\min_{\Delta} \max_{a} \Ev_{S\sim \Delta} \left[\tfrac1n|a \cg S|\right]  =  \max_{\Delta_a} \min_{S} \Ev_{a\sim \Delta_a} \left[\tfrac1n|a \cg S|\right] = \max_{\Delta_a} \min_{\Delta} \Ev_{\substack{a\sim \Delta_a\\S\sim \Delta}} \left[\tfrac1n|a \cg S|\right].$$
Above $\Delta_a$ denotes a distribution over candidates.
As a result, we can now allow $\Delta$ to depend on $\Delta_a$. Suppose we simply choose $\Delta = \Delta_a^k$ (sample $k$ times i.i.d.\@ from $\Delta_a$). Then we have
$$\Ev_{\substack{a\sim \Delta_a\\S\sim \Delta_a^{k}}} \left[\tfrac1n|a \cg S|\right] =  \Ev_{a_1, \dots, a_{k+1}\sim \Delta_a} \left[\tfrac1n|a_1 \cg a_2, \dots, a_{k+1}|\right].$$
By symmetry, this is at most $\frac{1}{k + 1}$ for any $\Delta_a$. (Note, the inequality can be made strict by considering the likelihood that $a_1$ is duplicated among $a_2, \dots, a_{k+1}$, which makes $\tfrac1n|a_1 \cg a_2, \dots, a_{k+1}|$ become 0.)
\end{proof}

As an aside, we remark that \Cref{cor:stab} already resolves the randomized analogue of \Cref{q:main}. Specifically, we can say that a distribution over committees of size $k$ is \textit{$\alpha$-undominated in expectation}, if for all candidates $a$, $\displaystyle\Ev_{S\sim \Delta}\big[\tfrac1n |a \cg S|\big] < \alpha$.  \Cref{cor:stab} shows that there always exists such a distribution $\Delta$ that is $\frac{1}{k + 1}$-undominated in expectation. This bound is tight. In an election with $m \gg k$ candidates with every possible ranking in equal proportion, it is not hard to see that the best distribution is uniform over committees of size $k$. Any candidate is preferred over a committee drawn from this distribution by a $\frac{1}{k + 1}(1 - \frac{k}{m}) \to \frac{1}{k + 1}$ fraction of the voters in expectation. Note that considering $k = 1$ (where the stable lottery is the distribution from the voting rule \textit{maximal lotteries} \cite{kreweras1965aggregation,fishburn1984probabilistic}), we get $\alpha = \frac12$, meaning that Condorcet's paradox does not extend to randomized voting rules.

\vspace{10pt}

The key difference in \Cref{lem:low} is that we want a guarantee over all sufficiently small subsets of voters $U$. To account for this, we change the formulation of the game so that the attacker chooses \textit{both} a candidate $a$ and a set $U$, and only voters in $U$ contribute to the payoff. We also have the function $g$ to contend with, but this can be accommodated with relative ease thanks to the generality of the minimax theorem.

\begin{proof}[Proof of \Cref{lem:low}]
For ease of exposition, we briefly switch to using notation more common in the optimization treatment of the minimax theorem. That is, we use vectors $\bm{x}, \bm{y}$ to denote probability distributions. These take the place of (but are not exactly the same as) $\Delta_a$ and $\Delta$.

Let $\bm{x}$ be a joint distribution over candidates and sets $U$ of voters. In particular we have non-negative weights $x_{a, U}$ for each candidate $a$ and subset $U \subseteq V$ such that $|U| \leq \alpha n$, and these sum to 1. (In what follows, $U$ and $U'$ always represent subsets of at most $\alpha n$ voters.) Let $\bm{y}$ be a distribution over candidates, which assigns probability $y_b$ to candidate $b$. Consider the function
$$f(\bm{x}, \bm{y}) :=  \frac1n \sum_{v \in V} \sum_{a \in C}\sum_{U: v\in U} x_{a, U}\cdot g\Bigg(\bigg(\sum_{b:a\cg_v b} y_b\bigg)^k\Bigg).$$

Note that $g\left(\left(\sum_{b:a\cg_v b} y_b\right)^k\right)$ is precisely $g(\rank_v(a))$ when $\Delta$ is $\bm{y}^k$ (the distribution over committees of size at most $k$ obtained by sampling $k$ times i.i.d.\@ from $\bm{y}$), since
$$\rank_v(a) = \Prx_{S\sim \bm{y}^k}[a\cg_v S] = \Prx_{b\sim \bm{y}}[a\cg_v b]^k = \bigg(\sum_{b:a\cg_v b} y_b\bigg)^k.$$
With this choice, observe as well that if $\bm{x}$ is set to $\bm{e}_{a, U}$ (deterministically picking some candidate $a$ and set $U$), we have 
$$ f(\bm{e}_{a, U}, \bm{y}) = \frac1n \sum_{v\in U} g(\rank_v(a)).$$
Therefore, in order to show the desired guarantee for any choice of $a, U$, it suffices to show that
$$\min_{\bm{y}} \max_{a, U} f(\bm{e}_{a, U}, \bm{y}) = \min_{\bm{y}} \max_{\bm{x}} f(\bm{x}, \bm{y}) \leq \int_{1 - \alpha}^1g(x^k) \dif x.$$
In order to apply the minimax theorem (\Cref{thm:minimax}), we need that the domains for $\bm{x}$ and $\bm{y}$ are compact convex, and $f(\cdot, \cdot)$ is convex in the minimizer's variables (each $y_b$) and concave in the maximizer's variables (each $x_{a, U}$). The first condition is satisfied since the domains are just probability distributions over discrete sets. Since $g(x^k)$ is convex, $f$ is convex in each $y_b$, and $f$ is linear and therefore concave in each $x_{a, U}$. It follows that $f$ satisfies the conditions required for the minimax theorem and so
$$ \min_{\bm{y}} \max_{\bm{x}} f(\bm{x}, \bm{y}) =  \max_{\bm{x}} \min_{\bm{y}} f(\bm{x}, \bm{y}).$$
Like in the proof of \Cref{cor:stab}, we can choose $\bm{y}$ in terms $\bm{x}$ to give an upper bound. We choose $\bm{y}$ such that for each candidate $b$,
$$y_b = \sum_{U} x_{b, U}.$$
Using this choice, it follows that 
\begin{equation}\label{eq:mess}
 \max_{\bm{x}} \min_{\bm{y}} f(\bm{x}, \bm{y}) \leq \max_{\bm{x}} \frac1n \sum_{v \in V} \sum_{a \in C}\sum_{U: v\in U} x_{a, U}\cdot g\left(\Bigg(\sum_{b:a\cg_v b} \sum_{U'} x_{b, U'}\Bigg)^k\right).
\end{equation}
Let us fix $\bm{x}$ and bound the contribution of each voter separately. To begin with, let 
$$w(v) := \sum_{a \in C}\sum_{U: v\in U} x_{a, U} $$
denote the total weight associated with voter $v$. Observe that the average weight is at most $\alpha$:
$$\frac1n \sum_{v\in V} w(v) = \frac1n\sum_{v\in V} \sum_{U: v\in U} \sum_{a\in C} x_{a, U} = \sum_{a, U} x_{a, U}\cdot \tfrac1n|U| \leq \alpha.$$
We show the following bound on the contribution of each voter.
\begin{equation}\label{eq:voter-contribution-bound}
\sum_{a\in C}\sum_{U: v\in U} x_{a, U}\cdot g\left(\Bigg(\sum_{b:a\cg_v b} \sum_{U'} x_{b, U'}\Bigg)^k\right) < \int_{1 - w(v)}^1 g(x^k)  \dif x.
\end{equation}
To simplify a little bit, note that the arguments of $g$ do not depend on each $U$, so we can write the sum equivalently as
\begin{equation}\label{eq:voter-contribution}
\sum_{a\in C}\Bigg(\sum_{U: v\in U} x_{a, U}\Bigg)\cdot g\left(\Bigg(\sum_{b:a\cg_v b} \sum_{U'} x_{b, U'}\Bigg)^k\right).
\end{equation}
We can bound this expression from above by reinterpreting it as a Riemann sum. See \Cref{fig:riemann} for a helpful picture.

\begin{figure}[ht!] %
  \centering
  \includegraphics{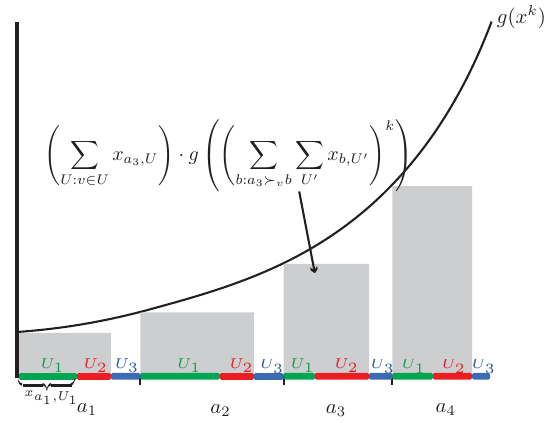} 
  \caption{A depiction of \Cref{eq:voter-contribution} as a Riemann sum. In the example, we suppose that $v \in U_1, U_2$. Note that the total width of the grey boxes is precisely $w(v)$. 
  } 
  \label{fig:riemann}
\end{figure}

Formally, suppose that $v$ has the preference order $a_1 \cl a_2 \cl \cdots \cl a_m$. Let 
$$z_i := \sum_{j = 1}^i \sum_{U} x_{a_j, U} \qquad \qquad \text{and} \qquad \qquad z_i^{(v)} := \sum_{j = 1}^i \sum_{U: v\in U} x_{a_j, U}.$$
Note that $z_m^{(v)} = w(v)$, and define $z_0^{(v)} := 0$. \Cref{eq:voter-contribution} can be written as 
$$\sum_{i = 1}^m \Big(z_i^{(v)} - z_{i - 1}^{(v)}\Big) g(z_i^k).$$
Using the fact that $g(x)$ is a non-decreasing function, we have
$$\sum_{i = 1}^m \Big(z_i^{(v)} - z_{i - 1}^{(v)}\Big) g(z_i^k) < \sum_{i = 1}^m \int_{z_i}^{z_i + z_i^{(v)} - z_{i - 1}^{(v)}} g(x^k) \dif x.$$
Note that the inequality is strict since equality can only occur if $g$ is a step function, which would contradict the assumption that $g(x^k)$ is convex and $g(x)$ is not constant. The right side considers the area under $g(x^k)$ over a subset of $[0,1]$ with total width $\sum_{i = 1}^m (z_i^{(v)} - z_{i - 1}^{(v)}) = z_m^{(v)} = w(v).$  Using the fact that $g$ is non-decreasing we can naively upper bound the total by the integral over $[1 - w(v), 1]$; the interval of length $w(v)$ in $[0, 1]$ where $g(x)$ is largest. 
It follows that
$$\sum_{i = 1}^m \Big(z_i^{(v)} - z_{i - 1}^{(v)}\Big) g(z_i^k) <  \int_{1 - w(v)}^{1} g(x^k) \dif x$$
which establishes \Cref{eq:voter-contribution-bound}. Returning to our average over the voters in \Cref{eq:mess}, all that remains is to show that
$$\frac1n \sum_{v \in V} \int_{1 - w(v)}^1 g(x^k) \dif x \leq \int_{1 - \alpha}^1 g(x^k) \dif x.$$
Define
$$p(t) := \int_{1 - t}^1 g(x^k) \dif x = \int_{0}^t g\big((1 - x)^k\big) \dif x.$$
Since $g$ is non-negative, $p(t)$ is non-decreasing. Moreover, since $g\big((1 - x)^k\big)$ is non-increasing, its integral $p(t)$ is a concave function. Therefore, by Jensen's inequality, we have that
$$\frac1n\sum_{v\in V} p(w(v)) \leq p\left(\frac1n\sum_{v\in V} w(v)\right) \leq p(\alpha) = \int_{1 - \alpha}^1g(x^k)  \dif x$$
as desired. 
\end{proof}

\subsection{\texorpdfstring{\Cref{lem:high}}{Lemma 3}: A Committee is High}\label{ssec:high}

Next, we prove \Cref{lem:high}, which shows that for any $\Delta$, there exists a committee $S$ whose ranks are high over any subset of $\alpha n$ voters. To do so, we first prove the following claim.

\begin{claim}\label{claim:high}
Fix any distribution $\Delta$ over committees. For each committee $S$, let $U_S \subseteq V$ denote the $\alpha n$ voters for whom $\rank_v(S)$ is smallest. Then
$$\Ev_{S\sim \Delta}\left[\frac1n\sum_{v \in U_S}g(\rank_v(S))\right] > \int_0^\alpha g(x)  \dif x.$$
\end{claim}

\Cref{lem:high} follows easily from \Cref{claim:high}. 

\begin{proof}[Proof of \Cref{lem:high} assuming \Cref{claim:high}]
By the probabilistic method, \Cref{claim:high} implies that there exists a committee $S$ in the support of $\Delta$ such that $\frac1n\sum_{v \in U_S}g(\rank_v(S)) > \int_0^\alpha g(x)  \dif x.$ But then, for any set $U \subseteq V$ of size at least $\alpha n$, since $g$ is non-decreasing and $U_S$ is the set of voters $v$ for whom $\rank_v(S)$ is minimal, we have
\[
\frac1n\sum_{v \in U}g(\rank_v(S))\geq \frac1n\sum_{v \in U_S}g(\rank_v(S)) > \int_0^\alpha g(x)  \dif x. \qedhere
\]
\end{proof}

\begin{proof}[Proof of \Cref{claim:high}]
First, we write the expectation in a way that we can isolate the contribution of each individual voter.
\begin{equation}\label{eq:expectation}
\Ev_{S\sim \Delta}\left[\frac1n\sum_{v \in U_S}g(\rank_v(S))\right] = \frac1n\sum_{v\in V} \sum_{S: v\in U_S} g(\rank_v(S))\cdot \Prx_\Delta(S).
\end{equation}
Let 
$$w(v) :=  \sum_{S: v\in U_S} \Prx_\Delta(S)$$
be the total probability mass associated with voter $v$. Observe that the average of $w(v)$ is $\alpha$, as
$$\frac1n \sum_{v \in V} w(v) = \frac1n \sum_{v \in V} \sum_{S: v\in U_S} \Prx_\Delta(S) = \sum_{S}\Prx_\Delta(S)\cdot \tfrac1n|U_S| = \alpha.$$
We bound the contribution of each voter in terms of $w(v)$. In particular, we show that for each $v$, 
\begin{equation}\label{eq:high-per-voter}
\sum_{S: v\in U_S} g(\rank_v(S))\cdot \Prx_\Delta(S) > \int_{0}^{w(v)} g(x) \dif x.
\end{equation}

\begin{figure}[ht!] %
  \centering
  \includegraphics{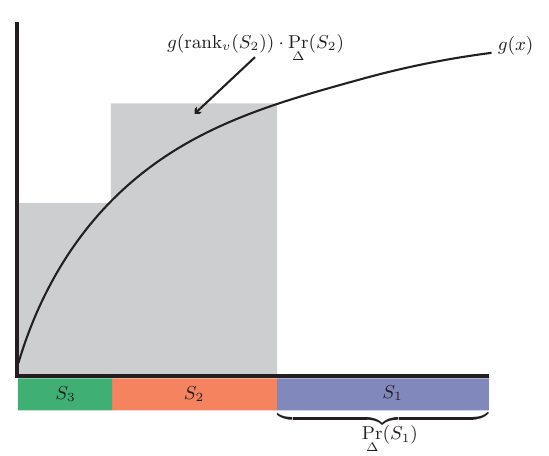} 
  \caption{A depiction of \Cref{eq:high-per-voter} as a Riemann sum. In the example, $v$ has the preference $S_1 \cgeq S_2 \cgeq S_3$, and $v \in U_{S_2}, U_{S_3}$. Note that the total width of the grey boxes is precisely $w(v)$. 
  } 
  \label{fig:riemann-high}
\end{figure}

The intuition is that like in the proof of \Cref{lem:low}, we can interpret the sum as a Riemann sum, and the worst case (minimizing the left side) is that all of the weight for committees $S$ such that $v \in U_S$ is concentrated in the bottom (see \Cref{fig:riemann-high}). 
More formally, since $g$ is non-decreasing,
$$\sum_{S: v\in U_S} g(\rank_v(S))\cdot \Prx_\Delta(S) > \sum_{S: v\in U_S} \int_{\rank_v(S) - \Prx_\Delta(S)}^{\rank_v(S)} g(x) \dif x.$$
The inequality is strict since equality can only occur if $g$ is a step function, which would contradict the assumption that $g(x^k)$ is convex and $g(x)$ is not constant.  The right side considers the area under $g(x)$ over a subset of $[0,1]$ with total width $\sum_{S: v\in U_S} \Prx_\Delta(S) = w(v).$  Using the fact that $g$ is non-decreasing we can naively lower bound the total by the integral over $[0, w(v)]$; the interval of length $w(v)$ in $[0, 1]$ where $g(x)$ is smallest. It follows that 
$$\sum_{S: v\in U_S} \int_{\rank_v(S) - \Prx_\Delta(S)}^{\rank_v(S)} g(x) \dif x \geq \int_{0}^{w(v)} g(x) \dif x$$
which establishes \Cref{eq:high-per-voter}. Returning to \Cref{eq:expectation}, we have shown that 
$$\Ev_{S\sim \Delta}\left[\frac1n\sum_{v \in U_S}g(\rank_v(S))\right] > \frac1n \sum_{v\in V} \int_0^{w(v)} g(x) \dif x.$$
The function $p(t) := \int_0^{t} g(x) \dif x$ is convex since $g(x)$ is non-decreasing. Then, by Jensen's inequality,
$$\frac1n \sum_{v\in V} \int_0^{w(v)} g(x) \dif x = \frac1n \sum_{v\in V} p(w(v)) \geq  p\left(\frac1n\sum_{v\in V} w(v)\right) = p(\alpha) = \int_{0}^\alpha g(x)  \dif x.$$
This proves the claim.
\end{proof}

\subsection{The Optimal Activation Function}\label{ssec:opt-g}

In this section, we explore what choice for the function $g$ gives us the strongest tradeoff between $\alpha$ and $k$. As we will see, the improvement over \Cref{thm:main} is marginal, but it demonstrates the best we can do with our techniques. 

Our goal is to find the non-constant non-decreasing function $g$ such that $g(x^k)$ is convex and

\begin{equation}\label{eq:condition-g}
\int_0^\alpha g(x)  \dif x \geq \int_{1 - \alpha}^1 g(x^k)\dif x
\end{equation}
is satisfied for the smallest possible $\alpha$. First, we claim that no such $g$ works if $\alpha^{1/k} < 1 - \alpha$. If this is true, then for all $x \leq \alpha$, we have $x \leq (1 - x)^k$. Using the fact that $g$ is non-decreasing, 
\[\int_0^\alpha g(x) \dif x \leq \int_0^\alpha g\big((1 - x)^k\big) \dif x = \int_{1 -\alpha}^1 g(x^k) \dif x.\]
Equality holds only if $g$ is constant, so this contradicts \Cref{eq:condition-g}. We note that this limitation implies that using \Cref{thm:meat}, we cannot show the existence of $\alpha$-undominated committees of size $k$ for $\alpha = o\big(\frac{\log k}{k}\big)$.

\vspace{10pt}

Henceforth, we assume that $\alpha^{1/k} \geq 1 - \alpha$. Using a change of variables $h(x) = g(x^k)$, \Cref{eq:condition-g} can be written equivalently as 
\begin{equation}\label{eq:condition-h-initial}
\int_0^{\alpha^{1/k}} kx^{k-1}h(x)  \dif x \geq \int_{1 - \alpha}^1 h(x) \dif x,
\end{equation}
which can be rearranged to 
\begin{equation}\label{eq:condition-h}
\int_{1-\alpha}^{\alpha^{1/k}} \big(kx^{k-1}-1\big) h(x)  \dif x \geq \int_{\alpha^{1/k}}^1 h(x) \dif x - \int_{0}^{1-\alpha} kx^{k-1} h(x) \dif x.
\end{equation}
The function $h$ must be non-constant, non-decreasing, and convex. We show that it suffices to consider functions of the form 
$$h(x) = \max(0, x - t)$$
for some parameter $t$. We call these \textit{shifted ReLU} functions. In particular, our goal is to show that if some arbitrary $h$ satisfies \Cref{eq:condition-h}, then we can use it to construct a shifted ReLU $\hat{h}$ that also satisfies \Cref{eq:condition-h}. We do this by making a series of transformations to $h$ that keep \Cref{eq:condition-h} invariant, summarized in \Cref{fig:best-h}.

\begin{figure}[ht!] %
  \centering
  \includegraphics{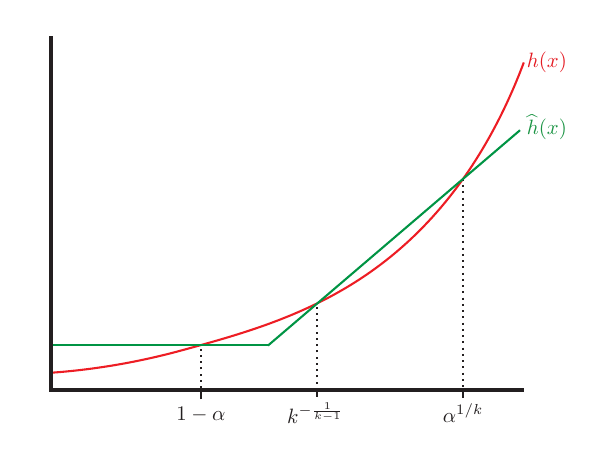} 
  \caption{The key transformation from $h$ to $\hat{h}$ that maintains \Cref{eq:condition-h}. After a linear transformation, the green curve becomes a shifted ReLU function.}
  \label{fig:best-h}
\end{figure}

Our first transformation is to replace $h$ with the function $$h_1(x) = \max(h(1 - \alpha), h(x)).$$
In \Cref{eq:condition-h}, this change only affects the last term on the right side. It increases $h$ on the interval $[0, 1 - \alpha]$, which in turn decreases the right side. Therefore if $h$ satisfies \Cref{eq:condition-h}, so does $h_1$. We can also observe that \Cref{eq:condition-h} is invariant to linear transformations; $h(x) \leftarrow h(x) + C$ or $h(x) \leftarrow C\cdot h(x)$. (This observation may be easier to see with \Cref{eq:condition-g}.) It follows that if $h$ satisfies \Cref{eq:condition-h}, so does $$h_2(x) = \max(0, h(x) - h(1 - \alpha)).$$

With this last transformation, we can assume that $h(x) = 0 $ for $x \leq 1 - \alpha$, and so \Cref{eq:condition-h} simplifies to 
\begin{equation}\label{eq:condition-h-simplfied}
\int_{1-\alpha}^{\alpha^{1/k}} \big(kx^{k-1}-1\big) h(x)  \dif x \geq \int_{\alpha^{1/k}}^1 h(x) \dif x.
\end{equation}
Intuitively, we would like $h(x)$ to be as small as possible in the interval $[\alpha^{1/k}, 1]$, and the subset of $[1-\alpha, \alpha^{1/k}]$ where $kx^{k-1} - 1$ is negative, and as large as possible in the subset of $[1-\alpha, \alpha^{1/k}]$ where $kx^{k-1} - 1$ is positive. Using the fact that $h$ is convex and non-decreasing, we can achieve this by replacing $h(x)$ on the interval $[1- \alpha, 1]$ with the maximum of a single line and 0, as follows.
\begin{itemize}
    \item If $k^{\frac{1}{1-k}} \geq \alpha^{1/k}$, replace $h$ with the line tangent to $h$ at $x = \alpha^{1/k}$.
    \item Otherwise, replace $h$ with the line passing through $\Big(k^{\frac{1}{1-k}}, h\big(k^{\frac{1}{1-k}}\big)\Big)$ and $\Big(\alpha^{1/k}, h\big(\alpha^{1/k}\big)\Big)$. 
\end{itemize}
This transformation decreases $h$ on the intervals $[1 - \alpha, k^{\frac{1}{1-k}}]$ and $[\alpha^{1/k}, 1]$ (which increase and decrease the left and right sides of \Cref{eq:condition-h} respectively), and increases $h$ on the (possibly empty) interval $[k^{\frac{1}{1-k}}, \alpha^{1/k}]$ (which increases the left side of \Cref{eq:condition-h}). The overall effect is that \Cref{eq:condition-h} is maintained.

With one last transformation scaling $h$ so that the line has slope 1, we have that $h$ is of the form $h(x) = \max(0, x - t)$ for some $t \in [1 - \alpha, \alpha^{1/k}]$ as desired.

\vspace{10pt}

It remains to determine which choice of $t$ allows the smallest $\alpha$ to satisfy \Cref{eq:condition-h}. For the shifted ReLU functions (with $t \in [1 - \alpha, \alpha^{1/k}]$), the equivalent \Cref{eq:condition-h-initial} simplifies once to
 $$ \int_{t}^{\alpha^{1/k}} kx^{k-1} (x - t)  \dif x \geq \int_t^1 (x - t)  \dif x.$$
The integral on the right is simply $\frac{(1 - t)^2}{2}$, whereas the integral on the left is
\begin{align*}
\left[\frac{k}{k+1}x^{k+1}\right]_{t}^{\alpha^{1/k}}   - t\left[x^k\right]_{t}^{\alpha^{1/k}} = \frac{k}{k+1}\alpha^{\frac{k+1}{k}} + \frac1{k+1} t^{k+1} - t\alpha.
\end{align*}
Therefore, the condition that we get is
\begin{equation}\label{eq:alpha-t-condition}
\frac{k}{k+1}\alpha^{\frac{k+1}{k}} + \frac1{k+1} t^{k+1} - t\alpha - \frac{(1 - t)^2}{2} \geq 0.
\end{equation}
If we fix $\alpha$, a sufficient choice of $t$ is the one that maximizes the expression on the left. Its derivative with respect to $t$ is precisely $-\alpha + t^k + (1 - t)$. 
It is not hard to check that at the boundaries where $t = 1- \alpha$ or $t = \alpha^{1/k}$, we get a stronger condition on $\alpha$ (and therefore a worse guarantee), so it suffices to consider $\alpha = t^k - t + 1$, where the derivative with respect to $t$ is 0. (Note as well that this choice is valid, since $t \in [1 - \alpha, \alpha^{1/k}]$ for $\alpha = t^k - t + 1$ and any $k \geq 1$.)
Substituting into the left side of \Cref{eq:alpha-t-condition}, we get the following condition on $t$
\begin{align*}
0 &\leq \frac{k}{k+1}( t^k - t + 1)^{\frac{k+1}{k}} + \frac1{k+1} t^{k+1} - t( t^k - t + 1) - \frac{(1 - t)^2}{2}\\
&= \frac{k}{k+1} \left((t^k - t + 1)^{\frac{k+1}{k}} - t^{k + 1}\right) - \frac{1 - t^2}{2}.
\end{align*}

Finally, we can conclude with the following theorem that gives the best tradeoff between $\alpha$ and $k$ that can be obtained with \Cref{thm:meat}.

\begin{theorem}\label{thm:stronger}
Suppose that  $k\geq 1$, and $t \in [0, 1]$ satisfies
$$\frac{k}{k+1} \left((t^k - t + 1)^{\frac{k+1}{k}} - t^{k + 1}\right) \geq \frac{1 - t^2}{2}.$$
then in any profile, there exists a committee of $k$ candidates that is $(t^k - t + 1)$-undominated.
\end{theorem}

\Cref{tab:best-alpha} gives the best bounds on $\alpha$ that can be obtained using \Cref{thm:stronger} for small choices of $k$. For comparison, we also include the simple lower bound of $\frac{2}{k + 1}$ and what \Cref{thm:main} gives. We can see that \Cref{thm:main} is a fairly high-fidelity approximation of \Cref{thm:stronger}. See \Cref{fig:comparison-plot} for a graphical comparison between our best upper bounds and the lower bounds.

\begin{table}[h!]
\centering
\begin{tabular}{c|c|c|c}
$k$ & Lower bound $(\alpha = \frac{2}{k+1})$ & \Cref{thm:stronger} & \Cref{thm:main} ($\frac{\alpha}{1 - \ln\alpha} = \frac{2}{k+1}$) \\ \hline
2 & 0.667      & 0.798134      & 0.808436      \\ \hline
3 & 0.5        & 0.673795      & 0.687411      \\ \hline
4 & 0.4        & 0.588554      & 0.602600      \\ \hline
5 & 0.333      & 0.525719      & 0.539214      \\ \hline
6 & 0.286      & 0.477066      & 0.489702     \\ \hline
7 & 0.25       & 0.438041      & 0.449760      \\ \hline
8 & 0.222       & 0.405896      & 0.416735      \\ 
\end{tabular}
\caption{Best bounds on $\alpha$ such that $\alpha$-undominated committees of size $k$ exists. For example, the entry in row $k = 6$ for column \Cref{thm:stronger} shows that every election has a $0.4771$-undominated set of size $6$.}
\label{tab:best-alpha}
\end{table}

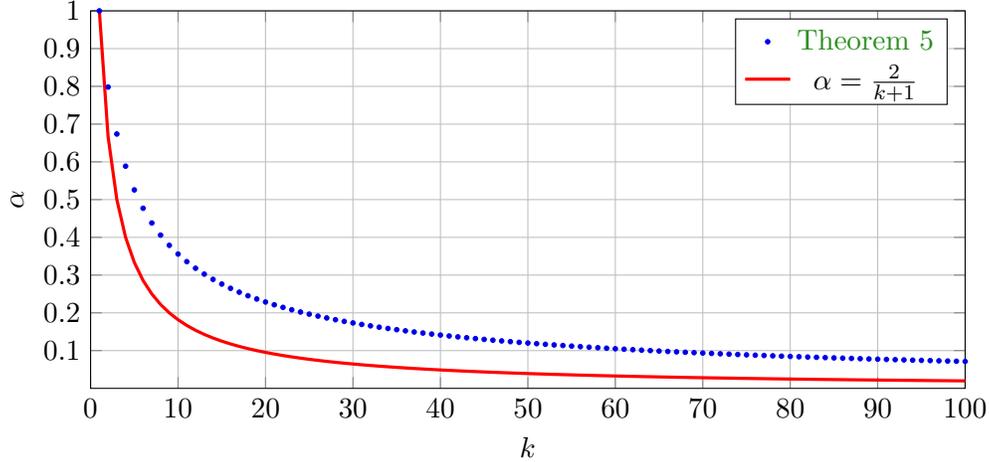
\begin{figure}
\centering
\begin{tikzpicture}%
    \begin{axis}[
        xlabel={$k$},               %
        ylabel={$\alpha$},               %
        xmin=0, xmax=100,       %
        ymin=0, ymax=1,       %
        ytick={0.1, 0.2, 0.3, 0.4, 0.5, 0.6, 0.7, 0.8, 0.9, 1},  %
        grid=both,                        %
        width=0.8\textwidth,                 %
        height=0.4\textwidth              %
    ]
        \addplot+[only marks, mark=*, mark size=0.85pt, color=blue] coordinates {
            (1, 1.0)
(2, 0.798134487238003)
(3, 0.6737945109704879)
(4, 0.5885539304003471)
(5, 0.5257193902904638)
(6, 0.4770664277484019)
(7, 0.4380411081092349)
(8, 0.40589630497820994)
(9, 0.3788651145522527)
(10, 0.35575337317625555)
(11, 0.33572157659422386)
(12, 0.31816031178709103)
(13, 0.30261525537743417)
(14, 0.2887400110773777)
(15, 0.27626534691755045)
(16, 0.26497849637683135)
(17, 0.25470885632003504)
(18, 0.245317864183352)
(19, 0.23669169512720223)
(20, 0.22873588887484342)
(21, 0.22137132958319705)
(22, 0.21453120493592237)
(23, 0.2081586642051293)
(24, 0.20220500055743085)
(25, 0.1966282169769631)
(26, 0.19139189339568385)
(27, 0.1864642748108355)
(28, 0.18181753690110036)
(29, 0.17742718463096696)
(30, 0.17327156420465495)
(31, 0.16933145224650814)
(32, 0.16558972181970988)
(33, 0.16203105870694667)
(34, 0.15864172534941912)
(35, 0.15540935590348848)
(36, 0.15232279302259555)
(37, 0.14937193594028053)
(38, 0.1465476183300195)
(39, 0.1438415013751504)
(40, 0.14124598222430018)
(41, 0.13875410994260307)
(42, 0.13635951838456073)
(43, 0.1340563666602982)
(44, 0.13183928029466585)
(45, 0.12970331016595438)
(46, 0.12764388760672274)
(47, 0.12565678869660513)
(48, 0.12373810184522815)
(49, 0.12188419938114692)
(50, 0.12009171139279073)
(51, 0.1183575034060188)
(52, 0.11667865450084414)
(53, 0.11505244248833701)
(54, 0.11347632274071484)
(55, 0.11194791649915659)
(56, 0.11046499609671256)
(57, 0.1090254754339306)
(58, 0.10762739561722512)
(59, 0.10626891495642277)
(60, 0.10494830501288799)
(61, 0.10366393668685958)
(62, 0.10241427483392751)
(63, 0.10119787384386147)
(64, 0.10001336655814819)
(65, 0.09885946255826128)
(66, 0.09773494231250268)
(67, 0.09663865053779186)
(68, 0.09556949385585656)
(69, 0.09452643426586704)
(70, 0.09350848945433421)
(71, 0.09251472636117064)
(72, 0.09154425647369424)
(73, 0.09059623796542027)
(74, 0.08966986851031222)
(75, 0.08876438493217875)
(76, 0.08787906032817905)
(77, 0.08701320035948101)
(78, 0.08616614554185043)
(79, 0.08533726501824879)
(80, 0.08452595461234846)
(81, 0.08373164175020065)
(82, 0.08295377413709693)
(83, 0.08219182727850238)
(84, 0.08144529745869877)
(85, 0.08071370243758702)
(86, 0.07999658130581255)
(87, 0.07929349091483084)
(88, 0.07860400746745899)
(89, 0.07792772382294255)
(90, 0.07726425024745576)
(91, 0.07661321173133284)
(92, 0.07597424789226459)
(93, 0.07534701202436034)
(94, 0.07473117273997554)
(95, 0.07412640958201344)
(96, 0.07353241467426275)
(97, 0.07294889179103126)
(98, 0.0723755554297617)
(99, 0.07181213161565292)
(100, 0.0712583543880394) %
        };
        \addlegendentry{\Cref{thm:stronger}}       %

        \addplot [
    domain=1:100, 
    samples=100, 
    color=red,
    style={very thick}
]
{2/(x+1)}; %
        \addlegendentry{$\alpha = \frac{2}{k + 1}$} %
\addplot [
    domain=1:100, 
    samples=100, 
    color=cyan,
    style={very thick}
]
{16/x};\addlegendentry{$\alpha = \frac{16}{k}$}
    \end{axis}
\end{tikzpicture}
    \caption{Comparison between the guarantee of \Cref{thm:stronger}, the lower bound, and the upper bound implied by \cite{jiang2020approximately} for $k\leq 100$.}
    \label{fig:comparison-plot}
\end{figure}

\section{Large Committees and Approximate Stability}\label{sec:core}
As we observed in the details of \Cref{ssec:opt-g}, \Cref{thm:meat} cannot be used to show that $\alpha$-undominated committees of size $k$ exist for $\alpha = o\big(\frac{\log k}{k}\big)$. In fact, there is a more fundamental reason for this, which is that the ultimate distribution $\Delta$ from which we sample our committee is a product distribution, and we sample from this distribution in a one-shot manner without any adaptivity. To give some intuition for why this gets stuck with a log factor, imagine that we have a large number of candidates with an equal number of voters with cyclic rankings over the candidates. Selecting a committee~$S$ and evaluating $\max_{a \in C} \frac1n |a\cg S|$ is essentially equivalent to choosing $k$ points on a circle of circumference 1, and evaluating the largest gap between adjacent points. Now, imagine that we sample the $k$ points i.i.d.\@ from some distribution. It is not hard to see that the distribution that minimizes the expected largest gap is the uniform distribution, and in that case, the expected largest gap grows as $\Theta\big(\frac{\log k}{k}\big)$.

A natural, more adaptive approach would be to sample just a small piece of the committee, assess how well it performs, and target the remaining committee to the voters where the initial committee is unsatisfactory. \cite{jiang2020approximately} do this by first sampling $(1 - \gamma)k$ candidates, and then recursing on the $\beta$ fraction of voters for whom $\rank_v(S)$ is the smallest (eventually with $\gamma = \frac12$ and $\beta = \frac14$).

In more detail, using \Cref{cor:stab} with $k \leftarrow (1 - \gamma)k$, we can sample from a distribution $\Delta$ such that for any candidate $a$, the total rank below $a$ is at most $\frac{1}{1 + (1 - \gamma)k} < \frac{1}{(1 - \gamma)k}$. On the other hand, when we sample $S \sim \Delta$, the expected fraction of voters for which $\rank_v(S) \leq \beta$ is at most $\beta$. Take a committee $S$ with this property, and let $W = \{v\in V: \rank_v(S) \leq \beta\}$. Considering the voters outside of $W$, we have that for any candidate $a$
$$\beta \cdot \tfrac1n |\{v\notin W: a\cg_v S\}| \leq \frac1n \sum_{\substack{v \notin W\\ a \cg_v S}} \rank_v(S) \leq  \frac1n \sum_{\substack{v \notin W\\ a \cg_v S}} \rank_v(a) < \frac{1}{(1 - \gamma)k}.$$
So, for any candidate $a$, the total number of voters that are outside $W$ and prefer $a$ over $S$ is at most $\frac{n}{\beta(1 - \gamma)k}$. Note that as a fraction of all of the voters, this is $O(1/k)$ when $\beta$ and $\gamma$ are constants.
Intuitively, we will treat this fraction as a sunk cost, and recurse on the voters in $W$. As the number of voters and committee sizes decrease exponentially, so do the costs incurred in each step. In the end, the cost of the first step will dominate, still giving us an $O(1/k)$ guarantee on the fraction of voters that prefers any $a$ over our total committee.

In more detail, when we recurse on the voters in $W$, the parameters
$k$ and $n$ are replaced with $\gamma k$ and $\beta n$. Following the same reasoning, after the next step, the maximum number of voters we give up on in the second step is at most $\frac{\beta \cdot n}{\gamma\cdot \beta(1 - \gamma)k}$. Continuing until the committee size becomes 1, the total number of voters that prefers $a$ over all of our committees is at most
$$\frac{n}{\beta(1 - \gamma)k} \sum_{t = 0}^{\infty}\left(\frac{\beta}{\gamma}\right)^t = \frac{n}{k} \cdot \frac{\gamma}{\beta(1 - \gamma)(\gamma - \beta)}.$$

This expression is minimized when $\gamma = \frac12$ and $\beta = \frac14$. Thus, for any candidate $a$, the fraction of voters that prefers $a$ over the union of committees we chose is at most $\frac{16}{k}$. %

It is natural to expect that building the committee in pieces adaptively is much better when the committee size is large, but unfortunately, it takes quite a while for the benefits to kick in. For example, we can see that we already lose a factor of $8$ in the first step of the recursion, largely due to the blunt binary evaluation of whether or not $\rank_v(S) \leq \beta$ (compared to the more fine-grained evaluation in \Cref{claim:high}). There is also the issue that for very small values of $k$, there is very little room to recurse. All of these reasons explain why there is a divergence of approaches for small and large values of $k$.

Nonetheless, we can show that some of the tools that we used in \Cref{sec:mainthm} can be leveraged to improve the asymptotic result of \cite{jiang2020approximately}.

\begin{reptheorem}{thm:stable}
In any election, there exists a $\frac{9.8217}{k}$-undominated committee of size $k$.
\end{reptheorem}

Note that using \Cref{prop:stable-undominated} (in \Cref{sec:equiv}), \Cref{thm:stable} implies the existence of $9.8217$-stable committees of size $k$ in the preference model of rankings over candidates. This result is only nontrivial  for $k \geq 10$, but it is enough to establish that Condorcet winning sets of size at most 20 always exist.

\begin{proof}
For ease of exposition, we replace $9.8217$ with a variable $c$, and the proof will impose constraints on $c$ that will be satisfied by $c = 9.8217$.
We will show by strong induction on $k$ that in any election, there exists a $\frac{c}{k}$-undominated committee of size $k$.

For the base case, we can use \Cref{thm:main} to establish the result for $k \leq 495$. (More straightforwardly, we could use the fact that the statement on its own is trivial for $k \leq 9$, but as we will see at the end of the proof, this would give us a less precise result.) Now suppose that $k \geq 496$ and that the result holds for all committee sizes $k' < k$.

Let $\beta, \gamma \in (0,1)$ be parameters to be set later. Like in \cite{jiang2020approximately}, we will first pick a committee $S$ of size $(1 - \gamma)k$, and then recurse on the $\beta$ fraction of voters which least like $S$. To use the tools in \Cref{sec:mainthm}, we will fix a function $g \colon [0,1]\to \mathbb{R}_{\geq 0}$ that is non-constant, non-decreasing, and such that $g\big(x^{(1 - \gamma)k}\big)$ is convex. Let $\Delta$ be the distribution obtained from \Cref{lem:low} when setting $k\leftarrow (1 - \gamma)k$ and $\alpha \leftarrow 1$. Then, we have that for all candidates $a$,
\begin{equation}\label{eq:stable-low}
\frac1n\sum_{v \in V} g(\rank_v(a)) < \int_{0}^1 g\Big(x^{(1 - \gamma)k}\Big)  \dif x.
\end{equation}
(We note that we could improve the right side by setting $\alpha \leftarrow \frac{c}{k}$ instead. As it turns out, the right choice of $g$ will have $ g\big(x^{(1 - \gamma)k}\big) =0$ for $x\leq 1-\frac{c}{k}$ anyway, so this weaker version is sufficient.)

On the other hand, since $g$ is non-decreasing, if we sample $S\sim \Delta$ the expected fraction of voters for which $g(\rank_v(S)) < g(\beta)$ is at most $\beta$. Fix some $S$ in the support of $\Delta$ with this property. Let $W = \{v\in V: g(\rank_v(S)) < g(\beta)\}$ and $U_a = \{v\in V\setminus W: a \cg_v S\}$ for each candidate $a$. For all $v \in U_a$, we have that $g(\rank_v(S)) \geq g(\beta)$, so combining with \Cref{eq:stable-low}, we get
$$g(\beta) \cdot \tfrac1n|U_a| \leq \frac1n \sum_{v \in U_a} g(\rank_v(S)) \leq \frac1n\sum_{v \in U_a} g(\rank_v(a)) < \int_{0}^1 g\Big(x^{(1 - \gamma)k}\Big) \dif x.$$
Therefore, the number of voters in $U_a$ is at most $n\cdot \frac1{g(\beta)}\int_{0}^1 g\big(x^{(1 - \gamma)k}\big) \dif x$. Applying the inductive hypothesis to the voters in $W$ with $k \leftarrow \gamma k$, we can choose a committee $S'$ of $\gamma k$ voters such that for any candidate $a$, at most 
$$\frac{c}{\gamma k} \cdot \beta n = \frac{\beta c}{\gamma k}\cdot n$$
voters in $W$ prefer $a$ over $S'$. Therefore, the fraction of voters that prefer $a$ over $S \cup S'$ is at most
$$\frac{\beta c}{\gamma k} + \frac1{g(\beta)}\int_{0}^1 g\Big(x^{(1 - \gamma)k}\Big) \dif x.$$
The induction is complete so long as we can choose $\beta, \gamma$, and $g$ such that 
$$\frac{\beta c}{\gamma k} + \frac1{g(\beta)}\int_{0}^1 g\Big(x^{(1 - \gamma)k}\Big) \dif x \leq \frac{c}{k}.$$
This inequality can be written equivalently as
$$c \geq \frac{\gamma k}{(\gamma - \beta)g(\beta)}\int_{0}^1 g\Big(x^{(1 - \gamma)k}\Big) \dif x.$$
Ultimately, we would like for the right side to be as small as possible. Since the conditions on $g$ are invariant to scaling, we can normalize so that $g(\beta) = 1$. As shorthand, let $\hat{\beta} = \beta^{\frac{1}{(1 - \gamma)k}}$. If we let $h(x) = g\big(x^{(1 - \gamma)k}\big)$, then $h$ is a non-constant, non-decreasing, convex function such that $h(\hat{\beta}) = 1$. What remains is to determine the best choice of $h$ with these properties.

It is not hard to see that if we take any arbitrary $h$, and replace it with the shifted ReLU that is tangent to $h$ at $x = \hat{\beta}$, then this reduces the value of $h(x)$ for each $x \in [0, 1]$ and thereby reduces $\int_{0}^1 h(x) \dif x$. Therefore, the best choice $h$ is of the form 
$$h(x) = \frac{\max(x - t, 0)}{\hat{\beta} - t}$$
for some $0 \leq t < \hat{\beta}.$ Then we have
$$\int_{0}^1 h(x) \dif x = \int_{t}^1 \frac{x - t}{\hat{\beta} - t} \dif x = \frac{(1 - t)^2}{2(\hat{\beta} - t)}.$$
It is not hard to check that at $t = 2\hat{\beta} - 1$, this expression attains its minimum value of $2(1 - \hat{\beta})$. Therefore, with this choice of $h$ (and thereby $g$), it suffices to find $\beta$ and $\gamma$ such that
$$c \geq \frac{\gamma k}{(\gamma - \beta)} \cdot 2\left(1 -  \beta^{\frac{1}{(1 - \gamma)k}}\right).$$
We can bound 
$$\beta^{\frac{1}{(1 - \gamma)k}} = \exp\left(\frac{\ln \beta}{(1 - \gamma)k}\right) \geq 1 + \frac{\ln \beta}{(1 - \gamma)k} = 1 - \frac{\ln(1/\beta)}{(1 - \gamma)k}$$
so $1 - \beta^{\frac{1}{(1 - \gamma)k}} \leq \frac{\ln(1/\beta)}{(1 - \gamma)k}$. Therefore, it suffices to have 
$$c \geq \frac{2\gamma \ln(1/\beta)}{(\gamma - \beta)(1 - \gamma)}.$$
We can choose $\beta < \gamma$ such that the expression on the right is as small as possible. Taking derivatives, we find that it is optimal to set $\beta = \gamma^2$. With this substitution, we are left with
$$c\geq \frac{4 \ln(1/\gamma)}{(1 - \gamma)^2}.$$
The right side is a single variable expression that is easy to optimize numerically. It obtains its minimum value of $9.82163$ when $\gamma \approx 0.28467$ (and $\beta \approx 0.08103$).

However, strictly speaking, the only valid choices for $\gamma$ are those such that $\gamma k$ is an integer. To remedy this,  we can observe that $\frac{4 \ln(1/\gamma)}{(1 - \gamma)^2} \leq 9.8217$ for all $\gamma \in (0.2832, 0.2862)$, and there is a multiple of $\frac{1}{k}$ in this range for all $k \geq 334$. Since all $k \leq 495$ are covered by \Cref{thm:main}, this completes the proof of the theorem. 
\end{proof}

For a smoother way of interpolating between \Cref{thm:main,thm:stable}, one can view the recursive process from the perspective of dynamic programming. We can think of $k'$ as $\gamma k$, the fraction of the committee that we recurse with. Tracing the analysis in the proof of \cref{thm:stable}, for each fixed choice of $k'$, the best choice of $\beta$ is $(\frac{k'}{k})^2$. If we allow $c$ to depend on $k$, then the proof gives us the recurrence
$$c(k) \leq \frac{k'}{k}\cdot c(k') + \frac{4\ln\frac{k}{k'}}{1 - \frac{k'}{k}}.$$
Minimizing over all choices of $k'$, $c(k)$ tends to $$\min_{\gamma \in (0, 1)} \frac{4 \ln(1/\gamma)}{(1 - \gamma)^2} \approx 9.82163.$$

Writing this recurrence in terms of $\alpha(k) = \frac{c(k)}{k}$, we get the following theorem.

\begin{theorem}\label{thm:dp}
Let $\alpha(k)$ be the infimum $\alpha$ such that every election has a committee of size $k$ that is $\alpha$-undominated. Then $\frac{\alpha(k)}{1 - \ln\alpha(k)} \leq \frac{2}{k+1}$, and
\[
\alpha(k) \leq \min_{1 \leq k' < k} \left(\left(\frac{k'}{k}\right)^2  \alpha(k') + \frac{4\ln \frac{k}{k'}}{k - k'}\right).
\]
\end{theorem}

Even still, we find that the first time that the recurrence improves \Cref{thm:main} is when $k = 300$.

\section{Discussion}\label{sec:discussion}

In elections where a Condorcet winner exists, selecting that candidate is generally seen as an uncontroversial choice. In our work, we show that for \cite{elkind2015condorcet}'s natural extension of Condorcet winners to committee selection, a small Condorcet winning set \textit{always} exists. We view this as a result that may be practically applicable. For instance, in awards like the Oscars or Grammys, it's possible to select a group of nominees such that only a few feel a contender was ``snubbed.'' Similarly, in faculty hiring, a department can select a shortlist for interviews without half the faculty urging the consideration of another candidate.

A second advantage of \cite{elkind2015condorcet}'s notion of $\alpha$-undominated sets is that it offers a more continuous measure for evaluating the quality of a committee. This flexibility stands in contrast to the binary axioms that are more typical in the social choice literature. \Cref{thm:main} offers a sliding scale of guarantees for different committee sizes $k$ and thresholds $\alpha$. 

We conclude with several exciting avenues for future work.

\paragraph{Closing the gap.} The obvious remaining question is to close the gap between the upper bounds (\Cref{thm:main,thm:stable}) and lower bounds (\Cref{thm:lb}). We conjecture that the lower bounds are optimal.

\begin{conjecture}
If $\alpha \geq \frac{2}{k+1}$ then $\alpha$-undominated sets of size $k$ always exist.   
\end{conjecture}

This conjecture would imply that every election has Condorcet dimension at most 3. One particularly interesting special case is when $k = 2$, where gap lies between $\frac23$ and $0.798$. Even to prove an upper bound of $0.99$, we are surprised that the simplest proof we know still goes through the minimax theorem. Finding a more elementary proof would be a natural first step. 

We also note that one reason \textit{not} to believe the conjecture is that there are several different ways to construct lower bounds. Those constructed by \cite{elkind2015condorcet,geist2014finding} are given in \Cref{tab:dim3,tab:dim3-6}. It is plausible that more sophisticated techniques, such as those borrowed from number theory and additive combinatorics to resolve \Cref{q:graph}, could be leveraged to produce stronger lower bounds.

\paragraph{Beyond ranked preferences.} It would be interesting to study whether the techniques we developed can be applied to study approximate stability for other monotone preference structures over committees, such as approval ballots. An intuition is that the ranking setting we studied in this paper is in a sense the hardest, since just one candidate is needed for any blocking committee $S'$. 
We suspect that some of the tools that we developed, especially the use of the activation function $g$, may be useful in other settings as well. %

\paragraph{Algorithmic directions.} Since our proofs are non-constructive, a natural question arises: are there efficient and practical algorithms for finding Condorcet winning sets (or more generally, $\alpha$-undominated sets)? One somewhat unsatisfying answer is that since $\alpha$-undominated sets are constant in size for reasonable values of $\alpha$ (e.g., at most $6$ for $\alpha = \frac12$), the brute-force algorithm is polynomial time. On the other hand, \cite{elkind2015condorcet} show that determining whether a given election has an $\alpha$-undominated set of size $k$ is $\mathsf{NP}$-complete. In our view, what would be most compelling for practical applications would be a simple, intuitive algorithm for reasonable values of $\alpha$ (like $\frac12$). We leave this as a direction for future research.

\anonymize{
\section*{Acknowledgments}
The authors would like to thank Bernhard von Stengel and William Zwicker for helpful discussions. Thanks to Pei-Yang Wu \edit{and the anonymous STOC reviewers for many}{} helpful comments.

Part of this work was done at the Workshop on Fairness in Operations Research at the Bellairs Research Institute. We thank the institute for its hospitality. 

Moses Charikar is supported by a Simons Investigator Award. Prasanna Ramakrishnan is supported by Moses Charikar’s Simons Investigator Award and Li-Yang Tan’s NSF awards 1942123, 2211237, 2224246, Sloan Research Fellowship, and Google Research Scholar Award. Adrian Vetta is supported by
NSERC Discovery Grant 2022-04191.
}

\bibliography{ref}
\bibliographystyle{alpha}

\appendix

\section{Equivalence between \texorpdfstring{$c$}{c}-stable and \texorpdfstring{$\frac{c}{k}$}{c/k}-undominated committees}\label{sec:equiv}

In this section, we formally establish the connection between $c$-stable committees (\Cref{def:approx-stable}) and $\frac{c}{k}$-undominated committees.

Recall that in the stable committee setting, voters have preferences over committees that are induced by their preferences over candidates. In particular, a voter $v$ \textit{strongly prefers} $S'$ over $S$ (denoted $S'\cg_v S$) if and only if $v$ prefers her favorite candidate in $S'$ over her favorite in $S$. 

We will borrow the notation that $\frac{1}{n}|S'\cg_v S|$ is the fraction of voters that strongly prefers $S'$ over $S$. Then restating \Cref{def:approx-stable}, a committee $S$ of size $k$ is \emph{$c$-stable} if for any committee $S'$ of size $k'$, $\frac{1}{n}|S'\cg_v S| \leq c\cdot\frac{k'}{k}$.

\begin{proposition}\label{prop:stable-undominated}
A committee $S$ of size $k$ is $c$-stable if and only if it is $\frac{c}{k}$-undominated.
\end{proposition}

\begin{proof}
Suppose that $S$ is a $c$-stable committee. Considering singletons $S' = \{a\}$, we have that for all candidates $a$, $\frac{1}{n}|a\cg S| < \frac{c}{k}$. This means precisely that $S$ is $\frac{c}{k}$-undominated.

Now, suppose that $S$ is \textit{not} a $c$-stable committee. Then there exists some $S'$ such that $\frac{1}{n}|S'\cg S| \geq c\cdot\frac{k'}{k}$. Among the voters for which  $S'\cg_v S$, let $a$ be the most frequent favorite in $S'$. Then it follows that $\frac{1}{n}|a\cg_v S| \geq \frac{c}{k}$, 
which means precisely that $S$ is not $\frac{c}{k}$-undominated.

Thus, the two notions are equivalent.
\end{proof}

Note that the existence of $c$-stable committees implies an upper bound of $2c$ on the Condorcet dimension of any elections.

\section{Lower bounds}\label{sec:lbs}

In this section, we include some negative results for the existence of $\alpha$-undominated sets. \Cref{tab:dim3} is the election given by \cite{elkind2015condorcet} with Condorcet dimension 3, which can be thought of as the Kroneker product of a $3$-cycle and a $5$-cycle.

\begin{table}[ht!]
    \centering
\begin{tabular}{ccccccccccccccc}
\hline
$v_1$ & $v_2$ & $v_3$ & $v_4$ & $v_5$ & $v_6$ & $v_7$ & $v_8$ & $v_9$ & $v_{10}$ & $v_{11}$ & $v_{12}$ & $v_{13}$ & $v_{14}$ & $v_{15}$ \\
\hline
\textcolor{red}{1} & \textcolor{red}{2} & \textcolor{red}{3} & \textcolor{red}{4} & \textcolor{red}{5} & \textcolor{blue}{6} & \textcolor{blue}{7} & \textcolor{blue}{8} & \textcolor{blue}{9} & \textcolor{blue}{10} & \textcolor{darkgreen}{11} & \textcolor{darkgreen}{12} & \textcolor{darkgreen}{13} & \textcolor{darkgreen}{14} & \textcolor{darkgreen}{15} \\
\textcolor{red}{2} & \textcolor{red}{3} & \textcolor{red}{4} & \textcolor{red}{5} & \textcolor{red}{1} & \textcolor{blue}{7} & \textcolor{blue}{8} & \textcolor{blue}{9} & \textcolor{blue}{10} & \textcolor{blue}{6} & \textcolor{darkgreen}{12} & \textcolor{darkgreen}{13} & \textcolor{darkgreen}{14} & \textcolor{darkgreen}{15} & \textcolor{darkgreen}{11} \\
\textcolor{red}{3} & \textcolor{red}{4} & \textcolor{red}{5} & \textcolor{red}{1} & \textcolor{red}{2} & \textcolor{blue}{8} & \textcolor{blue}{9} & \textcolor{blue}{10} & \textcolor{blue}{6} & \textcolor{blue}{7} & \textcolor{darkgreen}{13} & \textcolor{darkgreen}{14} & \textcolor{darkgreen}{15} & \textcolor{darkgreen}{11} & \textcolor{darkgreen}{12} \\
\textcolor{red}{4} & \textcolor{red}{5} & \textcolor{red}{1} & \textcolor{red}{2} & \textcolor{red}{3} & \textcolor{blue}{9} & \textcolor{blue}{10} & \textcolor{blue}{6} & \textcolor{blue}{7} & \textcolor{blue}{8} & \textcolor{darkgreen}{14} & \textcolor{darkgreen}{15} & \textcolor{darkgreen}{11} & \textcolor{darkgreen}{12} & \textcolor{darkgreen}{13} \\
\textcolor{red}{5} & \textcolor{red}{1} & \textcolor{red}{2} & \textcolor{red}{3} & \textcolor{red}{4} & \textcolor{blue}{10} & \textcolor{blue}{6} & \textcolor{blue}{7} & \textcolor{blue}{8} & \textcolor{blue}{9} & \textcolor{darkgreen}{15} & \textcolor{darkgreen}{11} & \textcolor{darkgreen}{12} & \textcolor{darkgreen}{13} & \textcolor{darkgreen}{14} \\
\textcolor{blue}{6} & \textcolor{blue}{7} & \textcolor{blue}{8} & \textcolor{blue}{9} & \textcolor{blue}{10} & \textcolor{darkgreen}{11} & \textcolor{darkgreen}{12} & \textcolor{darkgreen}{13} & \textcolor{darkgreen}{14} & \textcolor{darkgreen}{15} & \textcolor{red}{1} & \textcolor{red}{2} & \textcolor{red}{3} & \textcolor{red}{4} & \textcolor{red}{5} \\
\textcolor{blue}{7} & \textcolor{blue}{8} & \textcolor{blue}{9} & \textcolor{blue}{10} & \textcolor{blue}{6} & \textcolor{darkgreen}{12} & \textcolor{darkgreen}{13} & \textcolor{darkgreen}{14} & \textcolor{darkgreen}{15} & \textcolor{darkgreen}{11} & \textcolor{red}{2} & \textcolor{red}{3} & \textcolor{red}{4} & \textcolor{red}{5} & \textcolor{red}{1} \\
\textcolor{blue}{8} & \textcolor{blue}{9} & \textcolor{blue}{10} & \textcolor{blue}{6} & \textcolor{blue}{7} & \textcolor{darkgreen}{13} & \textcolor{darkgreen}{14} & \textcolor{darkgreen}{15} & \textcolor{darkgreen}{11} & \textcolor{darkgreen}{12} & \textcolor{red}{3} & \textcolor{red}{4} & \textcolor{red}{5} & \textcolor{red}{1} & \textcolor{red}{2} \\
\textcolor{blue}{9} & \textcolor{blue}{10} & \textcolor{blue}{6} & \textcolor{blue}{7} & \textcolor{blue}{8} & \textcolor{darkgreen}{14} & \textcolor{darkgreen}{15} & \textcolor{darkgreen}{11} & \textcolor{darkgreen}{12} & \textcolor{darkgreen}{13} & \textcolor{red}{4} & \textcolor{red}{5} & \textcolor{red}{1} & \textcolor{red}{2} & \textcolor{red}{3} \\
\textcolor{blue}{10} & \textcolor{blue}{6} & \textcolor{blue}{7} & \textcolor{blue}{8} & \textcolor{blue}{9} & \textcolor{darkgreen}{15} & \textcolor{darkgreen}{11} & \textcolor{darkgreen}{12} & \textcolor{darkgreen}{13} & \textcolor{darkgreen}{14} & \textcolor{red}{5} & \textcolor{red}{1} & \textcolor{red}{2} & \textcolor{red}{3} & \textcolor{red}{4} \\
\textcolor{darkgreen}{11} & \textcolor{darkgreen}{12} & \textcolor{darkgreen}{13} & \textcolor{darkgreen}{14} & \textcolor{darkgreen}{15} & \textcolor{red}{1} & \textcolor{red}{2} & \textcolor{red}{3} & \textcolor{red}{4} & \textcolor{red}{5} & \textcolor{blue}{6} & \textcolor{blue}{7} & \textcolor{blue}{8} & \textcolor{blue}{9} & \textcolor{blue}{10} \\
\textcolor{darkgreen}{12} & \textcolor{darkgreen}{13} & \textcolor{darkgreen}{14} & \textcolor{darkgreen}{15} & \textcolor{darkgreen}{11} & \textcolor{red}{2} & \textcolor{red}{3} & \textcolor{red}{4} & \textcolor{red}{5} & \textcolor{red}{1} & \textcolor{blue}{7} & \textcolor{blue}{8} & \textcolor{blue}{9} & \textcolor{blue}{10} & \textcolor{blue}{6} \\
\textcolor{darkgreen}{13} & \textcolor{darkgreen}{14} & \textcolor{darkgreen}{15} & \textcolor{darkgreen}{11} & \textcolor{darkgreen}{12} & \textcolor{red}{3} & \textcolor{red}{4} & \textcolor{red}{5} & \textcolor{red}{1} & \textcolor{red}{2} & \textcolor{blue}{8} & \textcolor{blue}{9} & \textcolor{blue}{10} & \textcolor{blue}{6} & \textcolor{blue}{7} \\
\textcolor{darkgreen}{14} & \textcolor{darkgreen}{15} & \textcolor{darkgreen}{11} & \textcolor{darkgreen}{12} & \textcolor{darkgreen}{13} & \textcolor{red}{4} & \textcolor{red}{5} & \textcolor{red}{1} & \textcolor{red}{2} & \textcolor{red}{3} & \textcolor{blue}{9} & \textcolor{blue}{10} & \textcolor{blue}{6} & \textcolor{blue}{7} & \textcolor{blue}{8} \\
\textcolor{darkgreen}{15} & \textcolor{darkgreen}{11} & \textcolor{darkgreen}{12} & \textcolor{darkgreen}{13} & \textcolor{darkgreen}{14} & \textcolor{red}{5} & \textcolor{red}{1} & \textcolor{red}{2} & \textcolor{red}{3} & \textcolor{red}{4} & \textcolor{blue}{10} & \textcolor{blue}{6} & \textcolor{blue}{7} & \textcolor{blue}{8} & \textcolor{blue}{9} \\
\hline
\end{tabular}
    \caption{An election with Condorcet dimension 3.}
    \label{tab:dim3}
\end{table}

The theorem below is a straightforward generalization of \Cref{tab:dim3}. 

\begin{theorem}[\cite{mathoverflow-lb}]\label{thm:lb}
If $\alpha < \frac{2}{k + 1}$, then there exist elections with no $\alpha$-undominated committees of size $k$.
\end{theorem}

\begin{proof}

Let $t$ be a large positive integer. We will generalize the instance in \Cref{tab:dim3}, multiplying a $(k + 1)$-cycle by a $t$-cycle (instead of a 3-cycle and 5-cycle). 

More formally, we will construct an election with $n = m = (k + 1) t$ voters and candidates. Let $[N] = \{0, 1, \dots, N - 1\}$ denote the residue classes modulo $N$, treated with the order $0 < 1 < \cdots < N - 1$.

Each voter and each candidate is associated with a unique tuple in $[k + 1]\times [t]$. The voter $(p, q)$ ranks the candidates in reverse lexicographical order after shifting the first entry down by $p$ and the second entry down by $q$ (modulo $k + 1$ and $t$ respectively). That is, 
$$(x, y) \cg_{(p, q)} (x', y') \quad \text{ if } \quad x - p  < x' - p \text{ or } (x = x' \text{ and } y - q  < y' -  q).$$ 

Now we claim that for any committee $S$ of $k$ candidates, there exists another candidate that is preferred by a $\frac{2}{k+1}\big(1 - \frac1t\big)$ fraction of the voters.

By the pigeonhole principle, there must exist $x \in [k + 1]$ such that $S$ contains at most one element of the form $(x, \cdot)$ and no element of the form $(x - 1, \cdot)$. If no $(x, \cdot)$ is in $S$, then $(x, 0)$ is preferred over $S$ by all of the $\frac{2}{k+1}$ fraction of voters of the form $(x-1, \cdot)$ and $(x, \cdot)$. If there is some $(x,y) \in S$, then $(x, y - 1)$ is preferred over $S$ by all of the $\frac{2}{k+1}\big(1 - \frac1t\big)$ voters of the form $(x-1, y')$ and $(x, y')$ for $y'\neq y$.  

To conclude, for any $\alpha < \frac{2}{k+1}$, there exists $t$ such that $\alpha \leq \frac{2}{k+1}\big(1 - \frac1t\big)$. Applying the construction above for such a choice of $t$, the result follows.
\end{proof}

We remark that
\cite{geist2014finding} put forth a minimal election with Condorcet dimension $3$. \Cref{tab:dim3-6} is a slightly modified version of their instance, made to be more interpretable. One can think of it as duplicating two $3$-cycles and weaving them together.
 
\begin{table}[h]
\centering
\begin{tabular}{cccccc}
\hline
$v_1$ & $v_2$ & $v_3$ & $v_4$ & $v_5$ & $v_6$\\
\hline
\textcolor{red}{1} & \textcolor{red}{2} & \textcolor{red}{3} & \textcolor{blue}{4} & \textcolor{blue}{5} & \textcolor{blue}{6}\\
\textcolor{blue}{4} & \textcolor{blue}{5} & \textcolor{blue}{6} & \textcolor{red}{3} & \textcolor{red}{1} & \textcolor{red}{2}\\
\textcolor{red}{2} & \textcolor{red}{3} & \textcolor{red}{1} & \textcolor{blue}{6} & \textcolor{blue}{4} & \textcolor{blue}{5}\\
\textcolor{red}{3} & \textcolor{red}{1} & \textcolor{red}{2} & \textcolor{red}{1} & \textcolor{red}{2} & \textcolor{red}{3}\\
\textcolor{blue}{6} & \textcolor{blue}{4} & \textcolor{blue}{5} & \textcolor{red}{2} & \textcolor{red}{3} & \textcolor{red}{1}\\
\textcolor{blue}{5} & \textcolor{blue}{6} & \textcolor{blue}{4} & \textcolor{blue}{5} & \textcolor{blue}{6} & \textcolor{blue}{4}\\
\hline
\end{tabular}
\caption{A minimal election with  Condorcet dimension 3.}
\label{tab:dim3-6}
\end{table}

\end{document}